\documentclass[11pt]{article}     
\usepackage[utf8]{inputenc}
\usepackage{float}
\usepackage{graphicx}
\usepackage{amsthm}

\usepackage{latexsym}
\usepackage{amsmath}
\usepackage{amssymb}
\usepackage{amsfonts}
\usepackage{multirow}
\usepackage{amsthm}

\usepackage{enumerate}
\usepackage{graphics}
\usepackage{fullpage}
\usepackage{color}
\usepackage{comment}

\usepackage{bussproofs}

\usepackage{booktabs, tabularx}

\usepackage{xspace}
\usepackage{xcolor,colortbl}

\usepackage{mathrsfs}

\usepackage{tikz}
\usetikzlibrary{trees,arrows,automata,shapes,decorations}

\newtheorem{thm}{Theorem}[section]

\newtheorem{lem}[thm]{Lemma}
\newtheorem{cor}[thm]{Corollary}

\newtheorem{exam}[thm]{Example}
\newtheorem{prop}[thm]{Proposition}

\newcommand{\ot}{\leftarrow}

\newcommand{\RN}[1]{%
  \textup{\uppercase\expandafter{\romannumeral#1}}%
}

\newcommand{\bM}{{\mathbf M}}

\newcommand{\bS}{{\mathbf S}}

\newcommand{\bE}{{\mathbf E}}

\newcommand{\myComment}[1]{}

\newcommand{\sima}{\sim_a}

\newcommand{\simb}{\sim_b}

\newcommand{\simB}{\sim_B}

\newcommand{\simC}{\sim_C}

\title{Learning What Others Know}
\author{Alexandru Baltag\footnote{Institute for Logic, Language and Computation, University of Amsterdam, A.Baltag@uva.nl.} $\,\,$ \& Sonja Smets\footnote{Institute for Logic, Language and Computation, University of Amsterdam and
 Department of Information Science and Media Studies, University of Bergen, S.J.L.Smets@uva.nl.}}
\date{}
\begin{document}
\maketitle

\begin{abstract} We propose a number of powerful dynamic-epistemic logics for multi-agent information sharing and acts of publicly or privately accessing other agents' information databases. The static base of our logics is obtained by adding to standard epistemic logic \emph{comparative epistemic assertions}, that can express \emph{epistemic superiority} between groups or individuals, as well as a \emph{common distributed knowledge} operator (that combines features of both common knowledge and distributed knowledge). On the dynamic side, we introduce actions by which epistemic superiority can be acquired: ``sharing all one knows'' (by e.g. giving access to one's information database to all or some of the other agents), as well as more complex informational events, such as hacking. We completely axiomatize several such logics and prove their decidability.
\end{abstract}

\section{Introduction}

In this paper, we look at actions by which agents gain access to other agents' information databases, and thus can in principle learn everything known to those others, acquiring epistemic superiority over them. We assume that information is distributed in a number of information sources or `sites' (e.g. files, folders, data sets, websites, databases etc.) at a given time. Each source can be thought of as being associated with an
agent, either because it is the knowledge base of a real agent (natural or artificial), or because we think of the source itself as an abstract `agent' (possessing exactly the information that is locally stored at that site).

\medskip

We enrich static epistemic logic with two new ingredients: (1) \emph{comparative epistemic assertions} for individuals or groups, that can capture epistemic superiority (e.g. ``she knows all they know''); (2) a new modal operator for \emph{common distributed knowledge}, that generalizes the two standard notions of common knowledge and distributed knowledge. On this static base, we built communication logics obtained by adding various dynamic operators for information sharing, public or private accessing etc.

\medskip

An agent may gain access to a site, after which it can be assumed to instantly `read' all the information stored at that source. The `reading' agent gains access to a source either because it is granted such access by the source agent itself (by ``sharing'' her database, in which case it is natural to assume that the source `knows' it is being accessed), or because it somehow succeeded to illegally gain such access via e.g. hacking (in which case typically the source doesn't  know it's being accessed).\footnote{Although sometimes it does get to know it, either because the hacker publicizes all the stolen information, or because somehow the source agent is able to detect the hacking. Our account can deal with various such scenarios.} So a reading action can be public (when it is common knowledge that the information is visible to everybody), or semi-public (when it is accessible only to some agents, but it is common knowledge who has access and who doesn't), or fully private (when both the information and the access are unknown to outsiders). Multiple agents may simultaneously access multiple sources. After each such reading action, each reading agent knows everything that was known by its source agents.


\bigskip

There are various possible applications of this work to multi-agent information gathering: e.g. multi-body planning tasks in which sensed information from different bodies, each having its own sensors, is to be collected and acted upon in order to reach a goal state \cite{RN}; recommender systems collecting user-preferences from multiple sources in order to provide a meaningful recommendation; cryptographic communication, involving protocols in which agents share their public keys and others use them to send messages, but also attacks by intruders getting access to private keys; etc.

\medskip

In the context of information accountability, here is a concrete example from \cite{refAI}. The agents are internet users, including website owners who have control over their own website as well as web robots (or web-crawlers) who can extract information from those websites. Such web robots can be directed to the URL of specific website owners and can be used for different purposes, e.g. to index website content. Yet not all web crawlers are designed for legitimate purposes: e.g. they can also be used extract valuable information for e.g. spamming; in the worst case, they can gain access to all the private content of some users and hence gain `epistemic superiority' over them. Website owners can disallow robots to visit their website (e.g. via `robots.txt' (https://www.robotstxt.org/) website owners can use a file to give instructions to the web robots or they can directly block an IP address). Giving such access-restricting instructions is a ``semi-public'' action (in the technical sense of our paper):
the `/robots.txt' file is publicly available, hence what parts are under `no-access'-restriction is public information. Still, robots used by spammers or malware robots could actually ignore
these instructions. In practice, it can be hard to detect whether a user's site has been visited by a web robot, especially as existing detection-methods are far from waterproof. Thus, the need for the more general setting in section \ref{product} of our paper, e.g. actions by which different agents secretly and simultaneously gain access to others' sites (without the owners' knowledge).

\bigskip

The paper is structured as follows: section \ref{prelim} gives some background on epistemic logic. In section \ref{superiority} we add epistemic comparative assertions for groups, and give a complete axiomatization of the resulting logic. In section \ref{semi-public} we study public and semi-public sharing/reading actions, and axiomatize them in the absence of common knowledge operators.  Motivated by the problems posed by common knowledge, we generalize this notion in section \ref{Cd} (to ``common distributed knowledge''), provide a complete and decidable axiom system, and use it to axiomatize semi-public actions. The proofs are relegated to the Appendix. Finally, in section \ref{product}, we further generalize this work to arbitrary reading actions, giving an axiomatization, and ending with a Conjecture, which we plan to settle in a future journal version of this paper.

\setcounter{tocdepth}{2}

\section{Preliminaries}\label{prelim}

An \emph{epistemic model} $\bS= (S, \sima, \underline{\bullet})_{a\in A}$ consists of: a set $S$ of \emph{states}; a family of equivalence relations $\sim_a\subseteq S\times S$, labelled by \emph{agents} $a\in A$ coming from a finite set $A$, and denoting the respective agents' epistemic \emph{indistinguishability relations}; and a \emph{truth-assignment function}\footnote{This last component is just a dual presentation of the more standard \emph{valuation map} $\|\bullet\|: Prop \to {\mathcal P}(S)$. Indeed, given the truth-assignment map, we can define the valuation by putting $\|p\|:=\{s\in S: \underline{s}(p)=1\}$. And vice-versa: given the valuation, we can put $\underline{s}:=\{p\in Prop: s\in \|p\|\}$.}
$\underline{\bullet}: S\to 2^{Prop}$, mapping each state $s\in S$ to a truth-assignment $\underline{s}:Prop\to 2=\{0,1\}$ defined on a given set $Prop$ of \emph{atomic propositions} (and mapping each $p\in Prop$ to a truth value $\underline{s}(p)\in \{0,1\}$). For any group of agents $B\subseteq A$, we define two equivalence relations $\sim_B, \sim^B\subseteq S\times S$:
$$\sim_B \,\, :=\,\, \bigcap_{b\in B} \sim_b, \,\,\,\,\,\,\,\,\,\,\,\,\,\,\,
\sim^B \,\, :=\,\, (\bigcup_{b\in B} \sim_b)^*,$$
where, for any relation $R\subseteq S\times S$, we take $R^*$ to denote the reflexive-transitive closure of $R$.

One can now introduce, for each group $B\subseteq A$,
a \emph{distributed knowledge} operator $D_B\varphi$ as the Kripke modality\footnote{The Kripke modality $[R]$ for a binary relation $R\subseteq S\times S$ is defined by putting $s\models [R]\varphi$ iff we have $t\models\varphi$  for all the states $t\in S$ with $sRt$.} for $\sim_B$, and a \emph{common knowledge} operator $C_B\varphi$ as the Kripke modality for $\sim^B$. In this paper, \emph{individual knowledge} $K_a\varphi$ is defined as just an \emph{abbreviation} for $D_{\{a\}}$.\footnote{But see e.g. \cite{FHMV} for an alternative treatment, in which both $K$ and $D$ are primitive operators, with $K_a\varphi$ being only logically equivalent to $D_{\{a\}}\varphi$.}

The \emph{logic of distributed knowledge} $LD$ has as language the set of all formulas built recursively from atomic formulas $p\in Prop$ by using negation $\neg\varphi$, conjunction $\varphi\wedge \psi$ and distributed knowledge operators $D_B\varphi$ (for all groups $B\subseteq A$). The \emph{logic of distributed knowledge and common knowledge} $LDC$ is obtained by extending the language of $LD$ with common knowledge modalities $C_B\varphi$.
These logics are known to be decidable and have the finite model property. Table \ref{tb0} below includes complete proof systems $\mathbf{LDC}$ and $\mathbf{LD}$ for these logics:

\vspace{-0.3cm}
\begin{table}[h!]
\begin{center}
{\small
\begin{tabularx}{\textwidth}{>{\hsize=0.5\hsize}X>{\hsize=1.5\hsize}X>{\hsize=0.8\hsize}X}
\toprule
\textbf{(I)} & \textbf{Axioms and rules of classical propositional logic}
 \vspace{1mm} \ \\
 \textbf{(II)} &\textbf{$S5$ axioms and rules for distributed knowledge}:   \ \\
($D$-Necessitation)& From $\varphi$, infer $D_B\varphi$  \ \\
($D$-Distribution) & $D_B (\varphi \to \psi) \to (D_B\varphi \to D_B\psi)$  \ \\
(Veracity) & $D_B\varphi \to \varphi$   \ \\
(Pos. Introspection) & $D_B\varphi \to D_B D_B\varphi$   \ \\
(Neg. Introspection) & $\neg D_B\varphi \to D_B \neg D_B\varphi$  \vspace{1mm} \ \\
\textbf{(III)} & \textbf{Special axiom for distributed knowledge}:  \ \\
(Monotonicity) & $D_B\varphi\to D_C\varphi$, for all $B\subseteq C\subseteq A$ \ \\
\textbf{(IV)} & \textbf{Axioms and rules for common knowledge}: \ \\
($C$-Necessitation) & From $\varphi$, infer $C_B\varphi$  \ \\
($C$-Distribution) & $C_B(\varphi\to \psi)\to (C_B\varphi\to C_B\psi)$  \ \\
($C$-Fixed Point) & $C_B\varphi \to (\varphi\wedge \bigwedge_{b\in B} K_b C_B\varphi)$  \ \\
($C$-Induction) & $C_B (\varphi \to \bigwedge_{b\in B} K_b \varphi)\to (\varphi\to C_B\varphi)$ \vspace{1mm} \ \\
\bottomrule
\end{tabularx}
}
\end{center}
\vspace{-0.5cm}
\caption{The proof system $\mathbf{LDC}$. Individual knowledge is a defined operator $K_a \varphi := D_{\{a\}}\varphi$. The system $\mathbf{LD}$ is obtained by eliminating the axioms in group (IV).}\label{tb0}
\end{table}
\vspace{-0.5cm}
\begin{exam} \label{one} The drawing below represents an epistemic model $\bS$ with 4 atomic propositions $Prop=\{p,q,r,w\}$ and 3 agents $A=\{a, b, c\}$. The possible states are represented by circles, inside which we write all the atomic propositions that are \emph{true} at that state. By default, the missing ones are false, so this fully captures each state's truth assignment (e.g. the circle labelled $p$ represents a state at which $p$ is true, but $q$ and $r$ are false. The epistemic indistinguishability relations are represented as edges (``links'') labelled by the respective agent. Since all our epistemic models are assumed to be $S5$, all $\simB$ are equivalence relations; hence, we skip the loops, as well as some of the links that can be obtained by transitivity.
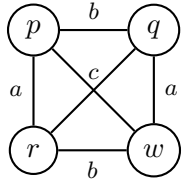
\vspace{-0.5cm}
\begin{center}
      \begin{tikzpicture}[node distance=1.6cm]
        \tikzstyle{zz}=[decorate,decoration={zigzag,post=lineto,post length=5pt}]

        \tikzstyle{w}=[draw=black,thick,circle,
        minimum size=1.6em]

        \tikzstyle{every edge}=[draw,thick,font=\footnotesize]

        \tikzstyle{every label}=[font=\footnotesize]

        \tikzstyle{ev}=[anchor=center,node distance=3.8cm]

        \tikzstyle{wred}=[w,draw=black]



        \node[w,label={left:},label={above:}] (w1) {$p$};

         \node[w, right of=w1,label={above:}] (w2) {$q$};

            \node[w, below of=w1,label={below:}] (w3) {$r$};

            \node[w, below of=w2,label={below:}] (w4) {$w$};






        \path (w1) edge[-] node[above]{$b$} (w2);

           \path (w1) edge[-] node[left]{$a$} (w3);

                \path (w1) edge[-] node[above]{$c$} (w4);

        \path (w2) edge[-] node[above]{$\,$} (w3);

                 \path (w2) edge[-] node[right]{$a$} (w4);

         \path (w3) edge[-] node[below]{$b$} (w4);

      \end{tikzpicture}
\end{center}
\vspace{-0.4cm}
\end{exam}
\noindent In this model, the disjunction of all atomic propositions is common knowledge: we have $C_{\{a,b,c\}} (p\vee q\vee r\vee w)$. In the $p$-state, $p$ is true, but this fact is not known to any individual agent. Instead, $p$ is distributed knowledge among all agents: we have $D_{\{a,b,c\}} p$. Intuitively, this distributed knowledge can be ``resolved'', i.e. converted into actual (common) knowledge, if the agents share all their information.\footnote{In \cite{SB}, we study different epistemic and doxastic states of groups of groups of agents that are realizable via specific sharing protocols.} In fact, in this state $p$ is distributed knowledge even within any $2$-agent group: we have $D_{\{a,b\}}p \wedge D_{\{b,c\}}p \wedge D_{\{a,c\}}p$. Again, intuitively this can be converted into common knowledge within each such $2$-agent group by using only in-group communication: e.g. if $a$ and $c$ tell each other all they know, then $C_{\{a,c\}}p$ holds after that. In fact, $a$ and $c$ become ``epistemically superior'' to $b$ after that: they will know all he knows.
Finally, note that in this model, the only way to obtain \emph{full} common knowledge $C_{\{a,b,c\}} p$ is to require \emph{every} agent to share her information with some others: no communication restricted to a specific $2$-agent subgroup can ever result in $C_{\{a,b,c\}} p$ in this model. As we'll see, this is \emph{not} the case in other models: very restricted forms of communication \emph{can} sometimes realize full common knowledge!

We are interested in extending the framework of epistemic logic to capture all the intuitive observations above. Standard temporal-epistemic logics \cite{FHMV, PaRam}, and dynamic approaches e.g. Public Announcement Logic (PAL) \cite{Plaza} and Dynamic Epistemic Logic (DEL) \cite{BMS,DHK,LDII,Baltag and Renne:2016}, can do this in a sense; but only by \emph{always making explicit the specific sentences} that are being communicated. This is not always convenient: the total sum of an agent's knowledge can typically be expressed only by a huge formula! In fact, sometimes this is worse: depending on the expressivity of the language, there might be \emph{no formula} in our language that captures this!

But even when there is one, there are problems with the standard DEL setting in some cases. In a purely syntactic approach to communication, the \emph{order} of the announcements matters: previously expressible information may become inexpressible after another announcement, which may prevent the full resolution of distributed knowledge \cite{Lonely}. Moreover, information that is locally expressible by formulas in every state may not be uniformly captured by any formula.\footnote{Say, all that agent $a$ knows is the value of some variable $x_a$ (ranging over natural numbers), e.g. some secret password. Suppose it is common knowledge (among all agents $a, b, c$) that $a$ shares this information with $b$. In each state, this is equivalent to a specific announcement of a sentence $x=n$ shared between $a$ and $b$. But from the perspective of the outsider agent ($c$), this is not equivalent to a specific announcement of \emph{any} sentence, and not even to \emph{any finite set} of possible such announcements! Indeed, to calculate $b$'s knowledge after this action in standard DEL, we need an event model with infinitely many events (one for each formula $x=n$ for any $n\in N$), all indistinguishable for agent $c$.}

What we need is to be able to abstract away from the specific announcement, and formalize directly the action of sharing ``all you know'' (with some or all of the other agents). Before doing that though, we need to formalize the \emph{effects} of such an action: the state of affairs in which one agent (or group) has \emph{epistemic superiority} over another agent (or group).

\section{We know all you know}\label{superiority}

As we saw, not all epistemic agents are equal. Some may know `more' than others: in fact, an agent $b$ may know \emph{everything} that another agent $c$ knows. This is easier and more realistic to assume if we identify agent $c$'s `knowledge' with the content of his associated information database. The more `expert' agent $b$ may have accessed this database, legally or illegally.

\medskip

In this paper, we extend epistemic logic $LDC$, with comparative epistemic assertions $B\preceq C$ between groups of agents $B, C\subseteq A$, saying that \emph{group $B$'s (distributed) knowledge includes all group $C$'s (distributed) knowledge}.\footnote{This is an extension to groups of the individual comparisons $b\preceq c$ in \cite{DitmarschEtAl}.} For short, we read this as: group $B$ ``knows at least as much'' as group $C$.
When $B\preceq C$ but $C\not\preceq B$, we write $B\prec C$ and say that group $B$ is ``more expert'' than (or ``epistemically superior to'') group $C$. As before, we skip set brackets when dealing with singletons, e.g. writing $b\preceq c$ for $\{b\}\preceq \{c\}$, etc.
The semantics is given by:
$$s\models B\preceq C \,\, \mbox{ iff } \,\, \forall t\in S \, (s\simB t \Rightarrow s\simC t).$$
This definition needs some explanation. Intuitively, the \emph{strongest} piece of knowledge collectively possessed by group $B$ at state $s$ (that entails \emph{everything known} by every $b\in B$) is $s$'s equivalence class $[s]_B=\{t\in S: s\simB t\}$ modulo $\simB$ (comprising all states compatible with the information possessed by agents in $B$).\footnote{Note that if this equivalence class shrinks, the knowledge of the agent (or group of agents) increases. The highest level of knowledge that an agent can achieve is the one in which she can distinguish between all states, i.e. when the equivalence classes are singletons. While this is a standard way of modelling knowledge in epistemic logic, philosophically this conception of knowledge is also well known in the literature and captures the concept of ``information as range'' \cite{HandbookJvB}.} The above clause says that $B\preceq C$ holds at $s$ iff $[s]_B\subseteq [s]_C$, i.e. if group $B$'s total distributed knowledge is at least as strong as group $C$'s distributed knowledge.

\medskip

\begin{exam} \label{oneprime} In the model in Example \ref{one}, group $\{a,b\}$ is `epistemically superior' to $c$: the distributed knowledge within $\{a,b\}$ includes everything known by $c$ but not the other way around (i.e. $\{a,b\}\preceq c$ but $c\not\preceq\{a,b\}$. In the same model, groups $\{a,b\}$ and $\{b,c\}$ are `epistemically equivalent': their distributed knowledge is the same (i.e. $\{a,b\}\preceq \{b,c\}$ and $\{b,c\}\preceq \{a,b\}$).
\end{exam}

\begin{exam} \label{onesecond} In the previous example, all mentioned epistemic comparisons hold globally (at all states). But in the model below, the group $\{a,c\}$ is epistemically superior to $\{b,d\}$ only in the $r$-state; dually, $\{b,d\}$ is superior to $\{a,c\}$ in the $q$-state; while in the $p$-state, the two groups are incomparable ($\{a,c\}\not\preceq \{b,d\}$ and $\{b,d\}\not\preceq \{a,c\}$. But groups $\{a,b\}$ and $\{c,d\}$ are epistemically equivalent ($\{a,b\}\preceq \{c,d\}$ and $\{c,d\}\preceq \{a,b\}$) in all states.
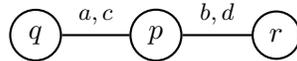
\vspace{-0.2cm}
\begin{center}
      \begin{tikzpicture}[node distance=1.6cm]
        \tikzstyle{zz}=[decorate,decoration={zigzag,post=lineto,post length=5pt}]

        \tikzstyle{w}=[draw=black,thick,circle,
        minimum size=1.6em]

        \tikzstyle{every edge}=[draw,thick,font=\footnotesize]

        \tikzstyle{every label}=[font=\footnotesize]

        \tikzstyle{ev}=[anchor=center,node distance=3.8cm]

        \tikzstyle{wred}=[w,draw=black]


        \node[w,label={below:}] (w1) {$q$};

        \node[w, right of=w1,label={below:},label={above:}] (w2) {$p$};

        \node[w,right of=w2,label={below:}] (w3) {$r$};


        \path (w1) edge[-] node[above]{$a,c$} (w2);

        \path (w2) edge[-] node[above]{$b,d$} (w3);

      \end{tikzpicture}
      \end{center}
\end{exam}

The following is our first new result, whose proof is sketched in Appendix \ref{CompletenessCd}.\footnote{The proof is rather intricate: both completeness and decidability involve a detour through a more general type of relational models, called pseudo-models.}
\begin{prop}\label{LCD-completeness}
The logic $LDC\preceq$, obtained by adding to the language of $LDC$ group comparison statements $B\preceq C$,
is decidable. A complete axiomatization is given by the proof system $\mathbf{LDC\preceq}$ in Table \ref{tb1}.
Moreover, the fragment $LD\preceq$ (obtained by eliminating the common knowledge operator) is axiomatized by the proof system $\mathbf{LD\preceq}$, obtained by removing from  Table \ref{tb1} the last group (IV) (the axioms and rules for common knowledge).
\end{prop}
\begin{table}[h!]
\begin{center}
{\small
\begin{tabularx}{\textwidth}{>{\hsize=0.5\hsize}X>{\hsize=1.5\hsize}X>{\hsize=0.8\hsize}X}
\toprule
\textbf{(I)} & \textbf{Axioms and rules of classical propositional logic}
 \vspace{1mm} \ \\
 \textbf{(II)} &\textbf{$S5$ axioms and rules for distributed knowledge}
\ \\&
(As in Table \ref{tb0})   
 \ \\
\textbf{(III)} & \textbf{Axioms for comparative knowledge}:  \ \\
(Inclusion) & $B\preceq C$, provided that $C\subseteq B$ \ \\
(Additivity) & $\left(B\preceq C\wedge B\preceq E\right) \to B\preceq C\cup E$  \ \\
(Transitivity) & $\left(B\preceq C\wedge C\preceq E\right)\to B\preceq E$ 
\ \\
 (Known Superiority) & $B\preceq C \to D_B (B\preceq C)$  \ \\
  (Knowledge Transfer) &  $B\preceq C \to \left(D_C\varphi \to D_B\varphi \right)$  \vspace{1mm} \ \\
  \textbf{(IV)} & \textbf{Axioms and rules for common knowledge}
  \ \\&
(As in Table \ref{tb0})
\ \\
\bottomrule
\end{tabularx}
}
\end{center}
\vspace{-0.6cm}
\caption{The proof system $\mathbf{LDC\preceq}$. Individual knowledge is a defined operator $K_a \varphi := D_{\{a\}}\varphi$. The system $\mathbf{LD\preceq}$ is obtained by eliminating the axioms in group (III).}\label{tb1}
\end{table}
Note that the axioms of group (III) take the place of the Monotonicity Axiom from Table \ref{tb0}, capturing natural properties of epistemic comparison and its interaction with distributed knowledge.\footnote{Indeed, Monotonicity becomes now \emph{provable} from these axioms.}
In particular, ``Known Superiority'' says that the more-expert group (collectively) knows its own epistemic superiority over a less-expert group. ``Knowledge Transfer'' says that
a more-expert group collectively knows everything known by a less-expert group.

\section{Tell me all you know: semi-public sharing}\label{semi-public}

We move on now to dynamics.
How can an agent $b$ come to know everything known by another agent $a$? One way is if $a$ actually shares all her information with $b$. In this section we assume this \emph{access permission is common knowledge}: all agents know that this access is being granted to $b$ (and know that the others know, etc). But note that we are not capturing $a$'s intentions or her deontic permissions, but only in the \emph{epistemic-informational features} of this action. For instance, suppose that $b$ gains access to $a$'s information without $a$'s permission (say, by hacking $a$'s information database), but this is done in such an obvious way that it is still common knowledge that it is being done (say, the hacker is `bragging': he issues a public statement confirming the hack). As long as $b$'s access gaining is still common knowledge, this information stealing has the same epistemic effect as the previously considered action of information sharing!

We can consider more general such actions, e.g. $a$ shares her information with a whole group $G$ (say, she gives permission to all agents in $G$ to access her knowledge base). Or all the agents in a group $H$ share all their information with another group $G$; or alternatively, some member of $G$ hacks $H$'s database, ``reads'' it and posts it all on a $G$-shared forum (so that all $G$-members can also ``read'' it), but the theft is discovered and publicly announced on TV; while, \emph{at the same time} another group $H'$ shares all \emph{their} information with group $G'$, etc.

We call all these actions \emph{semi-public `reading' events}. In all of them, some agents get to access (`read') some other agents' knowledge base(s). But \emph{the fact that this access is gained (or not) is public}: it is common knowledge who can ``read'' whose knowledge base during these events. The class of semi-public reading events include the \emph{fully public} ones, in which \emph{both the information} that is being accessed \emph{and the access itself are publicly available}: e.g. an agent or group publicly shares their information with \emph{everybody}; or when a hacker gains access to another agent database and posts on the internet all the information contained in it, thus making it all publicly available (cf. the WikiLeaks case).

\smallskip

\par\noindent\textbf{Reading maps}. To represent a semi-public reading event, we only need the specify who can ``read'' what.
A \emph{reading map} is a function $\alpha: A\to {\mathcal P}(A)$, mapping agents $a\in A$ to sets of agents $\alpha(a)\subseteq A$, subject to the constraint that
$$a\in \alpha(a) \,\, \mbox{ (for every $a\in A$).}$$
Intuitively, $\alpha(a)$ is the \emph{set of agents whose information is accessed by $a$} during this action. So this last constraint means that \emph{every agent $a$ can always re-read her own knowledge base}.\footnote{This is a technical assumption, not actually necessary (since we assume our agents have perfect memory, so they don't actually need to keep re-reading their own information), but which simplifies our reduction laws.}

Given a reading map $\alpha$, we extend the notation $\alpha(a)$ to groups of agents $B\subseteq A$, putting
$$\alpha(B) \,\, :=\,\, \bigcup_{b\in B} \alpha(b)$$
for the \emph{set of agents whose information can be accessed by some $B$-agent} during this action.

\medskip

\par\noindent\textbf{Notation conventions for reading maps}. In general, we denote specific reading maps by using tuples of expressions $a:B$, one for each agent $a$, to express the fact that agent $a$ reads the knowledge bases of all agents in $B$. So the tuple
$(a: G_a)_{a\in A}$ denotes the map $\alpha:A\to {\mathcal P}(A)$ given by $\alpha(a)=G_a$ for all $a\in A$. But we also introduce some conventions to simplify this notation: since $a\in \alpha(a)$ is assumed as a general condition, we can always choose to skip $a$ from the list of agents in $G_a$. Also, if $\alpha$ assigns the same reading assignment to two or more agents, we can compress the tuple, writing e.g. $G:H$ instead of the longer enumeration $(a: H)_{a\in G}$. Also, we skip the set brackets whenever either $G$ or $H$ is a singleton. It is also natural to be able to skip altogether from our tuple the agents $a$ who can only read their own base $\alpha(a)=\{a\}$. With these conventions, e.g. $(A:a)$ represents the map $\alpha$ given by $\alpha(b)=\{a,b\}$ for all $b\in A$; while $(G:H)$ represents the map $\beta$ given by: $\beta(b)=H\cup\{b\}$ if $b\in G$, and $\beta(b)=\{b\}$ otherwise.

\medskip

\par\noindent\textbf{Special reading maps}. We also introduce special notations for especially useful types of reading maps. Given a group $G$, we also ambiguously denote by $G$ the reading map $(A:G)$ (mapping every agent $a$ to $G\cup\{a\}$, so everybody reads the information possessed by $G$-agents).\footnote{We use systematic ambiguity: the reader can see from the context when $G$ denotes a group and when it denotes the corresponding reading map.} In particular, when $G=\{b\}$ is a singleton, we skip the set brackets as mentioned above, and write $a$ for the reading map $\{a\}=(A:a)$ (by which everybody reads $a$'s information). Finally, given mutually disjoint groups $G_1, \ldots, G_n$, we use the abbreviated notation $(G_1, G_2, \ldots, G_n)$ to denote the reading map $(G_1:G_1, G_2:G_2, \ldots , G_n:G_n)$ (that maps every agent $b$ to $G_k$ if $b\in G_k$ for some $k$, and to $\{b\}$ otherwise). As before, we skip set brackets when any of the $G_k$'s is a singleton. Note though that the reading maps $G$ and $(G)$ are \emph{different} (and the same for $a$ versus $(a)$).
In fact, this last notation can be naturally generalized to lists $G_1, \ldots, G_n\subseteq A$ of groups
that are \emph{not necessarily mutually disjoint}: this will denote the map $\alpha$ given by putting
$\alpha(b)=\bigcup\{G_k: 1\leq k\leq n \mbox{ with } b\in G_k\}$ if $b\in \bigcup_k G_k$, and $\alpha(b)=\{b\}$ otherwise.)

\medskip

We proceed now to formalize semi-public reading actions in DEL style
\cite{BMS, LDII, DHK}, as  \emph{epistemic updates}: functions mapping every epistemic model $\bS$ to a new model $\bS^{!\alpha}$.

\medskip

\par\noindent\textbf{Semantics of semi-public reading events}. Given a reading map $\alpha$, we denote by $!\alpha$ the corresponding semi-public event: \emph{it is common knowledge that every agent $a\in A$ simultaneously accesses the knowledge bases of all agents $b\in \alpha(a)$}. Formally, given any epistemic model $\bS= (S, \simb, \underline{\bullet})_{b\in A}$, the event $!\alpha$ returns an updated
model $\bS^{!\alpha}\,\, :=\,\, (S, \simb^{!\alpha}, \underline{\bullet})_{b\in A}$,
having the \emph{same set of states} $S$, the \emph{same valuation} $\underline{\bullet}$, but new epistemic indistinguishability relations $\simb^{!\alpha}$, given by:
$$\simb^{!\alpha} \,\, :=\,\, \sim_{\alpha(b)}$$
Intuitively, each agent $b$ acquires all the knowledge of group $\alpha(b)$, hence her new indistinguishability relation will coincide with the distributed knowledge relation for this group. (Note that, if $b$ has perfect memory, then his new knowledge relation should in fact be $\sim_b \cap \sim_{\alpha(b)}=\sim_{\{b\}\cup
{\alpha(b)}}$; but this is the same as $\sim_{\alpha(b)}$, given our above-mentioned simplifying assumption that $b\in \alpha(b)$.)

\medskip

\par\noindent\textbf{Adding dynamic modalities for semi-public reading actions}. As usual in Dynamic Epistemic Logic, we can now enrich the syntax of any of our logics by adding dynamic modalities $[!\alpha]\varphi$ for each reading map $\alpha$, saying that \emph{$\varphi$ will hold after the semi-public reading event $\alpha$}.
The \emph{semantic clause }for these dynamic modalities is given again as usual in Dynamic Epistemic Logic, by evaluating $\varphi$ at the same state in the updated model:
$$s\models_{\bS}[!\alpha]\varphi \,\,\, \mbox{ iff } \,\,\, s\models_{\bS^{!\alpha}}\varphi.$$

\begin{exam}\label{!a} (\emph{Tell Us All You Know}) For a given agent $b\in A$, $!b$ is a ``fully public'' action, formally given by the reading map $b=(A:b)$ (which according to the above conventions maps every $a \in A$ to $\{a,b\}$). This can be interpreted as public sharing: \emph{$b$ publicly announces all she knows}; but as already mentioned, it can also represent ``public hacking'': an anonymous hacker posts all $b$'s information on a public site. In the drawing below, we represent the effect of the action $!b$ performed on the initial epistemic model $\bS$ in Example \ref{one} (reproduced below in the diagram on the left). The updated model $\bS^{!b}$ after the action $!b$ is in the diagram on the right.
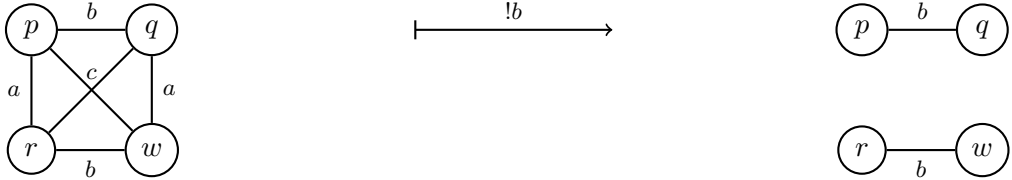
\vspace{-0.2cm}
\begin{center}
      \begin{tikzpicture}[node distance=1.6cm]
        \tikzstyle{zz}=[decorate,decoration={zigzag,post=lineto,post length=5pt}]

        \tikzstyle{w}=[draw=black,thick,circle,
        minimum size=1.6em]

        \tikzstyle{every edge}=[draw,thick,font=\footnotesize]

        \tikzstyle{every label}=[font=\footnotesize]

        \tikzstyle{ev}=[anchor=center,node distance=3.8cm]

        \tikzstyle{wred}=[w,draw=black]



        \node[w,label={left:},label={above:}] (w1) {$p$};

         \node[w, right of=w1,label={above:}] (w2) {$q$};

            \node[w, below of=w1,label={below:}] (w3) {$r$};

            \node[w, below of=w2,label={below:}] (w4) {$w$};




\node[right of=w2,node distance=8em,anchor=west] (ev0) {\,};
\node[right of=ev0,node distance=8em,anchor=east] (ev1){\,};

        \node[w,right of =ev1, node distance=8em, label={left:},label={above:}] (w5) {$p$};

         \node[w, right of=w5,label={below:}] (w6) {$q$};

            \node[w, below of=w5,label={below:}] (w7) {$r$};

            \node[w, below of=w6,label={below:}] (w8) {$w$};


        \path (w1) edge[-] node[above]{$b$} (w2);

           \path (w1) edge[-] node[left]{$a$} (w3);

                \path (w1) edge[-] node[above]{$c$} (w4);

        \path (w2) edge[-] node[above]{$\,$} (w3);

                 \path (w2) edge[-] node[right]{$a$} (w4);

         \path (w3) edge[-] node[below]{$b$} (w4);

           \path (w5) edge[-] node[above]{$b$} (w6);

                 \path (w7) edge[-] node[below]{$b$} (w8);

                   \path (ev0) edge[|->] node[above]{$!b$} (ev1);

      \end{tikzpicture}
\end{center}
\vspace{-0.3cm}
Before this communication event (i.e. in the model on the left), $a$ can distinguish between the left states and the right states (she knows
$p\vee r$ if the actual state is on the left, and knows $q\vee w$ if the actual state is on the right), $b$ can distinguish between the upper and the lower states, while $c$ can distinguish between the two diagonals. After $b$ publicly shares all his knowledge (i.e. in the updated model on the right), the other agents $a$ and $c$ will add $b$'s knowledge to their own, and will thus be able to distinguish between \emph{every} two states: they both come to \emph{know the actual state}. The only one still uncertain is $b$ himself (who learns nothing from his own announcement).
\end{exam}

\begin{exam}\label{!a} (\emph{Tell Me All You Know}) For given agents $a,b\in A$, the action $!(a:b)$ (given by the reading map $(a:b)$, according to the above conventions) is the one in which \emph{it is common knowledge that
$b$ shares with $a$ all she knows}. Note that this sharing event is \emph{not} fully public: the outsiders cannot read
$b$'s information (though they know that
$a$ can read it).
\end{exam}

\begin{exam}\label{!G} (\emph{You'all Tell Us All You Know}) For a given group $G\subseteq A$, the action $!G$ is the one by which \emph{all agents in $G$ publicly announce all they know}. Formally, it is given by the reading map $G$ (which maps every $b\in A$ to $G\cup\{b\}$). Like $!b$, this sharing action $!G$ is a ``fully public'' event. We illustrate this event in the diagram below. We start with the same initial epistemic model $\bS$ as in the previous example (on the left), and perform its update with the action $!\{a,b\}$, by which both agents $a$ and $b$ publicly share with everybody all they know. The result is the updated model $\bS^{!\{a,b\}}$ on the right.
After this action, \emph{everybody} comes to know the actual state, in fact if $p$ was the actual state them after this action \emph{$p$ becomes common knowledge} among all agents in $\{a,b,c\}$:
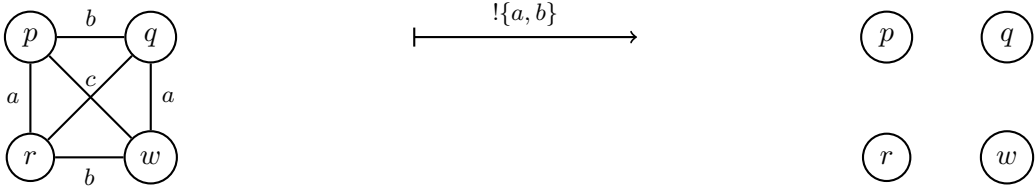
\vspace{-0.3cm}
\begin{center}
      \begin{tikzpicture}[node distance=1.6cm]
        \tikzstyle{zz}=[decorate,decoration={zigzag,post=lineto,post length=5pt}]

        \tikzstyle{w}=[draw=black,thick,circle,minimum size=1.6em]

        \tikzstyle{every edge}=[draw,thick,font=\footnotesize]

        \tikzstyle{every label}=[font=\footnotesize]

        \tikzstyle{ev}=[anchor=center,node distance=3.8cm]

        \tikzstyle{wred}=[w,draw=black]



        \node[w,label={left:},label={above:}] (w1) {$p$};

         \node[w, right of=w1,label={above:}] (w2) {$q$};

            \node[w, below of=w1,label={below:}] (w3) {$r$};

            \node[w, below of=w2,label={below:}] (w4) {$w$};





 \node[right of=w2,node distance=8em,anchor=west] (ev0) {\,};
 \node[right of=ev0,node distance=8em,anchor=west] (ev1) {\,};



        \node[w,right of =ev1, node distance=8em, label={left:},label={above:}] (w5) {$p$};

         \node[w, right of=w5,label={below:}] (w6) {$q$};

            \node[w, below of=w5,label={below:}] (w7) {$r$};

            \node[w, below of=w6,label={below:}] (w8) {$w$};



        \path (w1) edge[-] node[above]{$b$} (w2);

           \path (w1) edge[-] node[left]{$a$} (w3);

                \path (w1) edge[-] node[above]{$c$} (w4);

        \path (w2) edge[-] node[above]{$\,$} (w3);

                 \path (w2) edge[-] node[right]{$a$} (w4);

         \path (w3) edge[-] node[below]{$b$} (w4);

\path (ev0) edge[|->] node[above]{$!\{a,b\}$} (ev1);

      \end{tikzpicture}
\end{center}
\vspace{-0.3cm}
\end{exam}

\begin{exam}\label{!G} (\emph{Sharing Between Groups}) For a given group $G,H\subseteq A$, the semi-public reading action $!(G:H)$ (formally given by the reading map $(G:H)$) is the one in which \emph{it is common knowledge that all agents in $G$ share all they know with all agents in $H$}.

\end{exam}

\begin{exam}\label{!G} (\emph{Sharing Within Groups}) For groups $G_1, \ldots, G_n\subseteq A$, the semi-public event $!(G_1, \ldots, G_n)$ (given by the reading map $(G_1, \ldots, G_n)=(G_1:G_1, \ldots, G_n:G_n)$) is the one in which \emph{it is common knowledge that every agent in every group $G_k$ shares all she knows with the agents in that same group $G_k$}. A special case is the so-called \emph{$G$-resolution event} $!(G)$, in which (it is common knowledge) that agents in $G$ share all they know with each other.
Note the difference between $!(G)$ and the event $!G$ above. The corresponding dynamic operator $[!(G)]\varphi$ has already been considered in \cite{Agotnes}, under the name of \emph{resolution operator}, denoted by $R_G \varphi$.
\end{exam}

\smallskip
\par\noindent\textbf{Closure under sequential composition}. It is easy to see that \emph{the class of semi-public reading actions is closed under sequential composition}:
$$(\bS^{!\alpha})^{!\beta}= \bS^{!(\alpha\circ \beta)},$$
where $(\alpha\circ\beta)(a):= \alpha(\beta(a))$
is the \emph{functional composition} of reading maps. This immediately gives us the validity
$$[!\alpha][!\beta]\varphi \,\, \longleftrightarrow \,\, [!(\alpha\circ \beta)]\varphi,$$
known again as the ``Composition Law'' for public reading events.

\smallskip
\par\noindent\textbf{Subclasses closed (or not) under composition}.  Subclasses $!K$ of semi-public actions that are closed under sequential composition thus correspond to subclasses $K$ of reading maps that are closed under functional composition. An example is the class of group public sharing actions $\{!G: G\subseteq A\}$, which  is also closed under sequential composition. This can be easily seen from the fact that
$$(A:G)\circ (A:H)= (A:G\cup H),$$
which gives us
$$(\bS^{!G})^{!H}= \bS^{!(G\cup H)}.$$
In contrast, the class of individual sharing actions $\{!a: a\in A\}$ is \emph{not} closed under composition (since $!a;!b=!\{a,b\}\not = !c$ for any $c \in A$). Neither is the class $\{!(G): G\subseteq A\}$ of resolution actions, nor its extension to families of groups $\{!(G_1, \ldots, G_n): n\in N, G_1, \ldots, G_n\subseteq A\}$.

\smallskip
\par\noindent\textbf{The compositional closure of a class of actions} Given any subclass $!K$ of semi-public actions (based on a subclass $K$ of reading maps), we can look at its \emph{compositional closure}
$$!K^+ \,\, :=\,\, \{!(\alpha_1\circ \ldots \alpha_n): n\geq 1, \alpha_1, \ldots, \alpha_n\in K\},$$
which is the smallest class of actions that includes $K$ and is closed under sequential composition. For instance, it is easy to check that \emph{the compositional closure of the class $\{!(G): G\subseteq A\}$ of resolution actions} is the class
$$\{!(G_1\circ \ldots \circ G_n): n\geq 1, G_1, \ldots, G_n\subseteq A\}$$
where the reading map $(G_1\circ \ldots \circ G_n)\, :=\, (G_1)\circ\ldots \circ (G_n)$ is the functional composition of the reading maps $(G_1), \ldots, (G_n)$. It is useful to unfold this into the following explicit inductive definition of the above-defined liftings $\alpha(B)$ of these reading maps to sets of agents:
$$(G_1) (B) \, :=\,  G_1\cup B \,\,\,\, \mbox{ if $G_1\cap B\not=\emptyset$, } \,\,\,\,\,\,\, (G_1) (B)\, :=\, B \,\, \,\, \mbox{ otherwise; }$$
$$(G_1 \circ \ldots \circ G_n) (B)\, := \, G_1((G_2\circ\ldots \circ G_n) (B)),$$
from which we get their direct definition as maps from agents to sets:
$$(G_1 \circ \ldots \circ G_n) (b)\, := \, (G_1 \circ \ldots \circ G_n) (\{b\}).$$

\medskip
\par\noindent\textbf{Dynamic logics for semi-public actions}
The logic $LD\preceq!LD\preceq!$ is obtained by adding to the ``static'' language of $LD\preceq$ dynamic modalities $[!\alpha]\varphi$ for \emph{all} reading maps\footnote{Note that, if $A$ is finite, then there are only finitely many reading maps.} $\alpha$; while the logic $LDC\preceq!$ is obtained by adding such modalities to the language of $LDC$. Also, for any special class of reading maps, we can consider the logic with dynamic modalities restricted to the corresponding events, e.g. the ones of the form $!a$ (with $a\in A$), or $!G$ (with $G\subseteq A$), or $!(G)$.

\begin{prop}\label{LDsharing-completeness}
The dynamic logic  $LD\preceq!$ has the same expressivity as the static logic $LD\preceq$: every formula in $LD\preceq!$ is (provably) equivalent to a formula in $LD\preceq$ (via a step-by-step reduction using the Reduction Laws below). A complete axiomatization $\mathbf{LD\preceq!}$ of the dynamic logic $LD\preceq!$ is obtained by taking the axioms and rules of the system $\mathbf{LD\preceq}$ in Table \ref{tb1}, together the usual axioms and rules of normal modal logic\footnote{For details, see e.g. \cite{BRV}.} for dynamic modalities $[!\alpha]$, as well as the following \emph{`Reduction Laws' for semi-public reading actions}:
\par\noindent
{\centerline{
\begin{tabular}{llll}
$[!\alpha] p$ & $\,$ & $\longleftrightarrow$ $\,$ & $p$ \\
&&& \\
$[!\alpha] \neg \varphi$ & $\,$ & $\longleftrightarrow$ $\,$ & $\neg [!\alpha] \varphi$\\
&&& \\
$[!\alpha] (\varphi \wedge \psi)$ & $\,$ & $\longleftrightarrow$ $\,$ & $[!\alpha] \varphi\wedge [!\alpha] \psi$ \\
&&& \\
$[!\alpha] D_B \varphi$ & $\,$ & $\longleftrightarrow$ $\,$ & $D_{\alpha(B)} [!\alpha] \varphi,$ \\
&&&\\
$[!\alpha] (B\preceq C)$ & $\,$ & $\longleftrightarrow$ $\,$ & $\alpha(B)\preceq \alpha(C)$\\
\end{tabular}
}}
\end{prop}
\par\noindent The proof of this result is in Appendix \ref{reduction-completeness}.
In particular, applying the $D_B$-reduction law to a singleton group $B=\{b\}$, for any $b\in A$, we obtain a \emph{reduction law for individual knowledge}:
$$[ ! \alpha ] K_b \varphi  \, \, \, \, \, \, \, \Longleftrightarrow \, \, \, \, \, \, \, D_{\alpha(b)} [!\alpha] \varphi$$
Also, by restricting to the appropriate subclasses of events, we obtain axiomatizations of the corresponding logics. For instance, the logic of \emph{public sharing actions} $!G$ is axiomatized by the $!G$-instances of the reduction axioms, of which we only spell out the reductions for $D$ and $\preceq$:
$$[!G] D_B \varphi \, \, \,  \longleftrightarrow \, \, \,  D_{B\cup G} [!G] \varphi
\,\,\,\,\,\,\,\,\,\,\,\,\,\,\,\,\,\,\,\,\,\,\,\,\,\,
[!G] (B\preceq C) \, \, \,  \longleftrightarrow \, \, \,  B\cup G\preceq C\cup G$$

By restricting instead to \emph{``resolution'' actions} $!(G)$ (allowing sharing only within $G$), we obtain an
axiomatization for the \emph{logic of resolution} $!(G)$, in which e.g. the reduction law for $D$ splits in two cases, depending on the overlap of $G$ with $B$:\footnote{While the reduction law for $\preceq$ splits into \emph{four} cases, depending on the overlaps of $G$ with $B$ and with $C$.}
$$[!(G)] D_B \varphi  \, \, \, \longleftrightarrow  \, \,\,  D_{B\cup G} [!(G)] \varphi
\,\,\,\,\,\,\,\,\,\,\, \,\, \, \,\, \,\,  \,\,\,  \,\, \,\,\,  \,\,
\mbox{ for $B\cap G\not=\emptyset$ }$$
$$[!(G)] D_B \varphi \, \, \, \longleftrightarrow \, \,  \, D_{B} [!(G)] \varphi
\,\,\,\,\,\,\,\,\,\,\, \,\, \, \,\, \,\,  \,\,\,  \,\, \,\,\,  \,\, \,  \,\, \,\,  \,\,
\mbox{ for $B\cap G=\emptyset$ }$$
Here are some other interesting theorems of $\mathbf{LD\preceq!}$ :
$$[!a](b\preceq a)$$
(``After $a$ publicly announces all she knows, everybody comes to know all she knows.''),
$$D_{G\cup\{b\}} p \to [!G] K_b p$$
(``If $p$ is distributed knowledge in the group $G\cup\{b\}$, then after this group publicly shares all they know agent $b$ comes to know $p$.'')

\bigskip

But what about adding \emph{common knowledge} to this logic? The logic
$LDC\preceq!$, obtained by adding common knowledge operators $C_B$ to the language of $LD\preceq!$, can capture interesting fundamental properties. Here is a validity of $LDC\preceq!$ that involves the resolution operator $!(G)$ (by which agents $G$ share all they know with each other):
$$D_G p \leftrightarrow [!(G)]C_G p$$
(``Distributed knowledge $D_G p$ is the necessary and sufficient condition for realizing common knowledge $C_G p$ using only communication/sharing \emph{within the group} $G$).

\medskip

\par\noindent\textbf{No reduction for common knowledge} Unfortunately,
it turns out that \emph{there are no reduction laws for common knowledge} $[!G]C_B \varphi$ after sharing!
In the Dynamic Epistemic Literature, there are two well-known strategies for dealing with this problem. The \emph{first strategy}, used e.g. in \cite{BMS}, is to directly axiomatize the resulting logic, typically by using some kind of analogues of the axioms for common knowledge. The \emph{second strategy}, used e.g. in \cite{vBvEK}, is to enrich the static base of this logic with new operators, that allow for a reduction law for common knowledge.

\medskip

In the rest of this section, we \emph{sketch without proofs} the result of applying the first strategy to semi-public actions. Then in the next section, we systematically explore the second strategy.

\medskip

\par\noindent\textbf{A direct axiomatization of $LDC\preceq !$}
The idea of this first strategy is to treat the combination $[!\alpha]C_B$ \emph{as if} it was a single operator (``common knowledge after event $e$''), like a  kind of `dynamic version' of common knowledge. Then one can generalize the Fixed Point and Induction axioms to this dynamic combination.

To understand our generalization, it is convenient to first restate the $C$-Induction Axiom in terms of an \emph{inference rule}:
\vspace{-0.2cm}
$$\mbox{ From } \eta \to \left(\, \varphi\wedge \bigwedge_{b\in B} K_b \eta\, \right)  \, \, \, \, \, \, \, \mbox{ infer }  \, \, \, \, \, \, \, \eta \to C_B \varphi.$$
It is well-known (and easy to check) that this rule is equivalent to the Induction Axiom. We can now state our generalization, in the form of a dynamic \emph{$[!\alpha]C$-Induction Rule}:
\vspace{-0.2cm}
$$\mbox{ From } \eta \to \left(\, [!\alpha] \varphi\wedge \bigwedge_{b\in B} D_{\alpha(b)} \eta\, \right)  \, \, \, \, \, \, \, \mbox{ infer }  \, \, \, \, \, \, \, \eta \to [!\alpha] C_B \varphi.$$
The other ingredient we need is the above-mentioned Composition Axiom, which allows us to compress strings of dynamic modalities $[!\alpha_1]\ldots [!\alpha_n] C_B \varphi$ into a single dynamic modality $[!\alpha] C_B\varphi$.

\begin{prop}\label{LDCaction} A complete axiomatization of $LDC\preceq !$ consists of the following:
\begin{itemize}
\item the axioms and rules of the proof system $\mathbf{LD\preceq !}$;
\item the Axioms and Rules for Common Knowledge;
\item the Necessitation Rule and Distribution Axiom for dynamic modalities $[!\alpha]\varphi$;
\item the above Composition Axiom;
\item the above $[!\alpha]C$-Induction Rule.
\end{itemize}
\end{prop}

In fact, the proof is modular: given any class $!K$ of semi-public reading actions that is closed under sequential composition, we get a complete axiomatization of the logic $LDC\preceq !K$ (with dynamic modalities only for actions in $!K$) by restricting the above axioms only to the instances that belong to this logic.
As applications, we obtain the following two results:

\begin{cor}
A complete axiomatization of the logic $LDC\preceq !G$ of public sharing with common knowledge consists of the following:
\begin{itemize}
\item the above axiomatization of the logic $LD\preceq !G$ of public sharing actions;
\item the Axioms and Rules for Common Knowledge;
\item the Necessitation Rule and Distribution Axiom for dynamic modalities $[!G]\varphi$;
\item the $!G$-Composition Axiom: $[!G][!H]\varphi\leftrightarrow [!(G\cup H)]\varphi$;
\item the $[!G]C$-Induction Rule: \vspace{-0.2cm}
$$\mbox{ From } \eta \to \left(\, [!G] \varphi\wedge  D_{B\cup G} \eta\, \right)  \, \, \, \, \, \, \, \mbox{ infer }  \, \, \, \, \, \, \, \eta \to [!G] C_B \varphi.$$
\end{itemize}
\end{cor}

\begin{cor}
A complete axiomatization of the logic $LDC\preceq !(G)$ of resolution actions with common knowledge consists of the following:
\begin{itemize}
\item the above axiomatization of the logic $LD\preceq !(G)$ of resolution actions;
\item the Axioms and Rules for Common Knowledge;
\item the Necessitation Rule and Distribution Axiom for dynamic modalities $[!(G)]\varphi$;
\item the $[!(G)]C$-Induction Rule: \vspace{-0.2cm}
$$\mbox{ From } \eta \to \left(\, [!(G_1)]\ldots [!(G_n)] \varphi\wedge  D_{(G_1\circ\ldots G_n)(B)} \eta\, \right)  \, \, \, \, \, \, \, \mbox{ infer }  \, \, \, \, \, \, \, \eta \to [!(G)] C_B \varphi,$$
where $(G_1\circ\ldots G_n)(B)$ is the composed reading map defined by the inductive clauses above.
\end{itemize}
\end{cor}

The last result is a correction (and extension) of the system in \cite{Agotnes}, where the logic of resolution actions was studied, but only in the absence of comparative knowledge statements  $B\preceq C$. The
authors of \cite{Agotnes} gave reduction laws for distributed knowledge after resolution (the same as the ones we obtained by applying our general reduction laws to resolution events). They also proposed an axiomatization for the extension with common knowledge, based on a dynamic version of the Induction Rule, similar to our $[!(G)]C$-Induction Rule. However, their version of the rule is much simpler than ours, and seems to us to be ``wrong'': sound, but too weak to be complete.\footnote{The induction rule for $[!(G)]C$ in \cite{Agotnes} uses (conjunctions of) individual knowledge in the premise, instead of distributed knowledge, which is very strange
 (since the reduction law for $[!(G)]K$ uses $D$). In any case, the completeness proof for that rule in \cite{Agotnes} contains a \emph{gap}. Our $[!(G)]C$-Induction Rule looks so complicated because the class of resolution events is \emph{not} closed under composition.
}

\medskip

We relegate the proofs of these results to a future journal version of this paper, since they are not central to the line of investigation pursued here. In the next section, we solve the same problem using the ``second strategy'' mentioned above (following \cite{vBvEK}): enrich the static base of this logic with new operators, allowing for \emph{simple reduction laws} for common knowledge after \emph{any} semi-public event. The resulting axiomatization will \emph{not} rely on closure under composition (and hence will be applicable to \emph{every} subclass of semi-public actions).

\section{Common distributed knowledge}\label{Cd}

To `pre-encode' common knowledge after a public or semi-public reading action, we need to introduce a relatively new concept\footnote{As far as we know, this concept was first defined, but not axiomatized, in an ILLC master thesis \cite{Suzanne} supervised by the first author.}: \emph{common distributed knowledge}. Though motivated here by the aim of obtaining reduction laws, this epistemic notion is of interest in its own respect.

Given a family ${\mathcal B}\subseteq {\mathcal P}(A)$ of groups of agents, we say that $\varphi$ is said to be \emph{common distributed knowledge} among (the groups in) ${\mathcal B}$ if we have that: each group $B\in {\mathcal B}$ has distributed knowledge that $\varphi$; each group $B\in {\mathcal B}$ has distributed knowledge that each other group $B'\in {\mathcal B}$ has distributed knowledge that $\varphi$; etc (for all iterations). Formally:
$$s\models Cd_{\mathcal B}\varphi \,\, \mbox{ iff }\,\, s\models D_{B^1} D_{B^2}\ldots D_{B^n} \varphi \mbox{ for all sequences (of any length $n\geq 0$) } B^1, \ldots, B^n\in {\mathcal B}.$$
Equivalently, we can define $Cd_{\mathcal B}$ as the Kripke modality for the relation $\sim^{\mathcal B}$, given by
$$\sim^{\mathcal B}\,\, :=\,\, (\bigcup_{B\in {\mathcal B}}\sim_B)^*$$
(where as before, $R^*$ is the reflexive-transitive closure of $R$). Unfolding this definition, we get:
$$s\models Cd_{\mathcal B}\varphi \,\, \mbox{ iff }\,\, t\models \varphi \mbox{ holds at every state $t$ reachable by any finite chain}$$
$$\mbox{ (of any length $n\geq 0$) } s=s_0\sim_{B^1} s_1 \sim_{B^2}\ldots \sim_{B^n} s_n=t \mbox{ with all } B_i\in {\mathcal B}.$$
Here is one way to explain the informational significance of common distributed knowledge, versus plain distributed knowledge.
We already noted the validity $D_G p \, \longleftrightarrow \, [!\alpha^G] C_G p$, saying that distributed knowledge $D_G p$ is the sufficient and necessary precondition for realising common knowledge $C_G p$ by information sharing only \emph{within} the group $G$.
But given a family ${\mathcal B}=(G_1, \ldots, G_n)\subseteq {\mathcal P}(A)$ of groups of agents, the question arises: \emph{when can we achieve common knowledge of $p$ in the larger group $G=G_1\cup\ldots \cup G_n$
by info-sharing only within each of the subgroups} ($G_1, \ldots, G_n$)?

The answer is: whenever \emph{$p$ is common distributed knowledge} among the groups $B_1, \ldots, B_n$. This fact is captured by the validity
$$Cd_{G_1, \ldots, G_n} \,\, \longleftrightarrow \,\, [!(G_1, \ldots, G_n)] C_{G} p,$$
where $G=G_1\cup\ldots \cup G_n$
and $!(G_1, \ldots, G_n)$ is the semi-public event of sharing-within-each-group, as
 defined in the previous section, via the reading map $(G_1, \ldots, G_n)=(G_1:G_1, \ldots, G_n:G_n)$.

\smallskip

\begin{exam} \label{four} In the model drawn in Example \ref{onesecond} (reproduced below on the left), $p$ is common distributed knowledge in the $p$-state between groups $\{a,b\}$ and $\{c,d\}$, i.e. we have $Cd_{\{a,b\}, \{c,d\}}p$: all iterations of $D_{\{a,b\}} p$, $D_{\{c,d\}} p$, $D_{\{a,b\}} D_{\{c,d\}} p$, $D_{\{c,d\}} D_{\{a,b\}} p$ etc, hold at this state. This is witnessed dynamically by the fact that full common knowledge of $p$ can be achieved by sharing information only within the two groups, as witnessed by the drawing below: the updated model after $!(\{a,b\}, \{c,d\})$ is on the right-side, and its $p$-state satisfies $C_{\{a,b,c,d\}} p$.
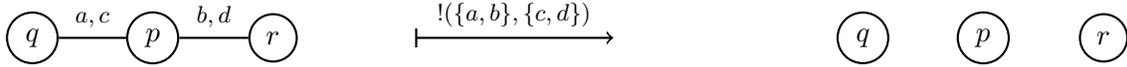
\begin{center}
      \begin{tikzpicture}[node distance=1.6cm]
        \tikzstyle{zz}=[decorate,decoration={zigzag,post=lineto,post length=5pt}]

        \tikzstyle{w}=[draw=black,thick,circle,
        minimum size=1.6em]

        \tikzstyle{every edge}=[draw,thick,font=\footnotesize]

        \tikzstyle{every label}=[font=\footnotesize]

        \tikzstyle{ev}=[anchor=center,node distance=3.8cm]

        \tikzstyle{wred}=[w,draw=black]



        \node[w,label={below:}] (w1) {$q$};

        \node[w, right of=w1,label={below:},label={above:}] (w2) {$p$};

        \node[w,right of=w2,label={below:}] (w3) {$r$};


        \path (w1) edge[-] node[above]{$a,c$} (w2);

        \path (w2) edge[-] node[above]{$b,d$} (w3);


  \node[right of=w2,node distance=8em,anchor=west] (ev0) {\,};
  \node[right of=ev0,node distance=8em,anchor=east] (ev1) {\,};



 \node[w,right of =ev1, node distance=8em, label={left:},label={above:}] (ev2) {$q$};

         \node[w,right of =ev2, label={left:},label={above:}] (w4) {$p$};

 \node[w,right of =w4, label={left:},label={above:}] (w5) {$r$};


\path (ev0) edge[|->] node[above]{$!(\{a,b\}, \{c,d\})$} (ev1);

      \end{tikzpicture}
\end{center}
\end{exam}

\begin{exam} \label{five}
In contrast, here is an example in which $p$ is distributed knowledge in each of the two groups, but it is \emph{not} common distributed knowledge. In the \emph{left-side model} of the diagram below, the upper $p$-state satisfies both $D_{\{a,b\}} p$ and $D_{\{c,d\}} p$ ; but we also have $\neg D_{\{a,b\}} D_{\{c,d\}} p$ in this world; hence $p$ is \emph{not} common distributed knowledge in the family $\{\{a,b\}, \{c,d\}\}$. This is witnessed by the fact that sharing within each the two groups cannot produce full common knowledge of $p$. Indeed, the action $!(\{a,b\}, \{c,d\})$ produces the \emph{right-side model}, in which we do \emph{not} have $C_{\{a,b,c,d\}} p$ in the upper $p$-state:
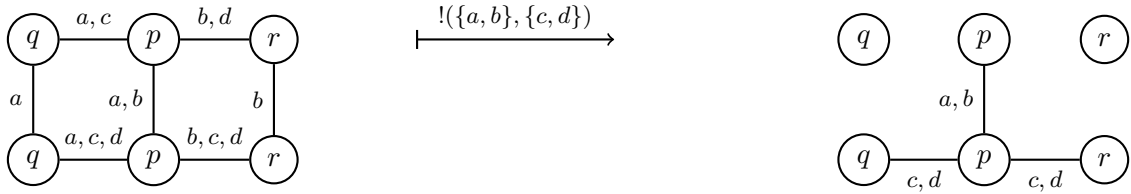
\begin{center}
      \begin{tikzpicture}[node distance=1.6cm]
        \tikzstyle{zz}=[decorate,decoration={zigzag,post=lineto,post length=5pt}]

        \tikzstyle{w}=[draw=black,thick,circle,
        minimum size=1.6em]

        \tikzstyle{every edge}=[draw,thick,font=\footnotesize]

        \tikzstyle{every label}=[font=\footnotesize]

        \tikzstyle{ev}=[anchor=center,node distance=3.8cm]

        \tikzstyle{wred}=[w,draw=black]


        \node[w,label={left:}] (w1) {$q$};

   \node[w, right of=w1,label={above left:},label={above:}] (w2) {$p$};

      \node[w,right of=w2,label={right:}] (w3) {$r$};

   \node[w, below of=w1,label={below:}] (w4) {$q$};

      \node[w, below of=w2,label={below:}] (w5) {$p$};

      \node[w, below of=w3,label={below:}] (w6) {$r$};

        \path (w1) edge[-] node[above]{$a,c$} (w2);

        \path (w2) edge[-] node[above]{$b,d$} (w3);

         \path (w4) edge[-] node[above]{$a,c,d$} (w5);

        \path (w5) edge[-] node[above]{$b,c,d$} (w6);

        \path (w1) edge[-] node[left]{$a$} (w4);

        \path (w2) edge[-] node[left]{$a,b$} (w5);

           \path (w3) edge[-] node[left]{$b$} (w6);



  \node[right of=w2,node distance=8em,anchor=west] (ev0) {\,};

  \node[right of=ev0,node distance=8em,anchor=east] (ev1) {\,};



   \node[w,right of =ev1, node distance=8em, label={left:},label={above:}] (ev2) {$q$};

         \node[w,right of =ev2,  label={left:},label={above:}] (w7) {$p$};

    \node[w,right of =w7,  label={left:},label={above:}] (ev3) {$r$};

         \node[w, below of=w7,label={below:}] (w8) {$p$};

            \node[w, left of=w8,label={below:}] (w9) {$q$};

            \node[w, right of=w8,label={below:}] (w10) {$r$};

        \path (w7) edge[-] node[left]{$a,b$} (w8);

        \path (w8) edge[-] node[below]{$c,d$} (w9);

         \path (w8) edge[-] node[below]{$c,d$} (w10);

\path (ev0) edge[|->] node[above] {$!(\{a,b\}, \{c,d\})$} (ev1);

      \end{tikzpicture}
\end{center}
\end{exam}

\medskip

\par\noindent\textbf{Static and dynamic logics} The \emph{static logic of common distributed knowledge $LCd\preceq$} has $Cd_{\mathcal B}$ as \textit{the only modalities} (one for each family ${\mathcal B}\subseteq {\mathcal P}(A)$), in addition to atomic propositions, Boolean connectives and comparative statements $B\preceq C$. Its \emph{dynamic counterpart} $LCd\preceq!$ has in addition dynamic modalities $[!\alpha]\varphi$, for all reading maps $\alpha$.

In these logics, \emph{all the standard epistemic
operators are definable as abbreviations}: $D_B \varphi \,:=\, Cd_{\{B\}}\varphi$,
$K_b \varphi \, := \, Cd_{\{\{b\}\}}\varphi$, $C_B \varphi  \, := \, Cd_{\{\{b\}: b\in B\}}\varphi$.

\begin{prop}\label{LCd}
The static logic $LCd\preceq$ is decidable. A sound and complete axiomatization is given by the proof system
$\mathbf{LCd\preceq}$ in Table \ref{tb2}.
\end{prop}

\begin{table}[h!]
\begin{center}
{\small
\begin{tabularx}{\textwidth}{>{\hsize=0.8\hsize}X>{\hsize=2\hsize}X>{\hsize=0.6\hsize}X}
\toprule
\textbf{(I)} & \textbf{Axioms and rules of classical propositional logic} \vspace{1mm} \ \\
 \textbf{(II)} & \textbf{Axioms and rules for common distributed knowledge}:   \ \\
($Cd$-Necessitation)& From $\varphi$, infer $Cd_{\mathcal B}\varphi$  \ \\
($Cd$-Distribution) & $Cd_{\mathcal B}(\varphi\to \psi)\to (C_{\mathcal B}\varphi\to C_{\mathcal B}\psi)$  \ \\
($Cd$-Fixed Point) & $Cd_{\mathcal B}\varphi \,\, \to \,\, (\varphi\wedge \bigwedge_{B\in {\mathcal B}} D_B Cd_{\mathcal B}\varphi)$  \ \\
($Cd$-Induction) & $Cd_{\mathcal B} (\varphi \to  \bigwedge_{B\in {\mathcal B}} D_B\varphi) \,\, \to \,\, (\varphi\to Cd_{\mathcal B}\varphi)\, $  \ \\
($Cd$ Neg. Introspection) & $\neg Cd_{\mathcal B}\varphi \to Cd_{\mathcal B} \neg Cd_{\mathcal B}\varphi$  \vspace{1mm} \ \\
\textbf{(III)} & \textbf{Axioms for comparative knowledge}  \ \\
 & (As in Table \ref{tb1}) \\
\bottomrule
\end{tabularx}
}
\end{center}
\vspace{-0.5cm}
\caption{The proof system $\mathbf{LCd\preceq}$. Distributed knowledge, common knowledge and individual knowledge are defined operators: $D_B\varphi:=Cd_{\{B\}}\varphi$, $C_B\varphi:=Cd_{\{\{b\}:b\in B\}}\varphi$, $K_b \varphi := Cd_{\{\{b\}\}}\varphi$.}\label{tb2}
\end{table}

\bigskip

The completeness and decidability proofs are included in Appendix \ref{CompletenessCd}. Once again, the proofs are intricate,involving a detour through a non-standard relational semantics.

\medskip

Note that the old axioms and rules for $D$ and $C$ are now both replaced by the axioms and rules for $Cd$ (group (II) in Table \ref{tb2}): indeed, one can easily check that those old axioms for $D$ and $C$ are now derivable in $\mathbf{LCd\preceq}$.

\bigskip
\bigskip

Finally, we obtain our desired axiomatization of $LCd\preceq!$:

\begin{prop}\label{LCd!}
The dynamic logic $LCd\preceq!$ has the same expressivity as its static base $LCd\preceq$. A complete axiomatization $\mathbf{LCd\preceq!}$ is obtained by putting together the axioms and rules of the proof system  $\mathbf{LCd\preceq}$ above with the ones of the proof system $LD\preceq!$
from Proposition \ref{LDsharing-completeness}, as well as with the following \emph{Reduction law for Common Distributed Knowledge}:
$$[!\alpha] Cd_{\mathcal B}\varphi  \, \, \, \, \, \, \, \longleftrightarrow \, \, \, \, \,\, \, Cd_{
\{\alpha(B): B\in {\mathcal B}\}} [!\alpha] \varphi$$ \vspace{-0.5cm}
\end{prop}
The proof of this result is in Appendix \ref{reduction-completeness}.
Note that the Reduction Law for distributed knowledge $D_B$ from Proposition \ref{LDsharing-completeness} is in fact redundant now: we can regain it by applying the reduction law for common distributed knowledge $Cd_{\mathcal B}$ to a singleton family ${\mathcal B}:=\{B\}$.

Once again, we can obtain axiomatizations of various sublogics, by restricting the above axioms to the appropriate classes of events: for instance, we get an axiomatization of the logic of fully public sharing $!G$ and common distributed knowledge, with the following reduction axiom for $Cd$:
$$[! G] Cd_{\mathcal B} \varphi \, \, \, \, \, \, \, \longleftrightarrow \, \, \, \, \, \, \, Cd_{\{B\cup G:B\in {\mathcal B}\}} [!G]\varphi.$$
In a similar way, we obtain an axiomatization of the logic of `resolution' actions $!(G)$ and $Cd$, in which the instances of the reduction law for $Cd$ split again in two cases:
$$[!(G)] Cd_{\mathcal B} \varphi \, \, \, \, \, \, \, \longleftrightarrow \, \, \, \, \, \, \, Cd_{B\cup G} [!(G)]\varphi \, \, \, \, \, \, \,\, \, \, \, \, \, \, \mbox{ for $B\cap G\not=\emptyset$},$$
$$[!(G)] Cd_{\mathcal B} \varphi \, \, \, \, \, \, \, \longleftrightarrow \, \, \, \, \, \, \, Cd_{B} [!(G)]\varphi \, \, \, \, \, \, \,\, \, \, \, \, \, \, \,\, \, \,\, \, \,\,\mbox{ for $B\cap G=\emptyset$}.$$

\section{Wilder scenarios: arbitrary reading events}\label{product}

Until now, our relational models captured only \emph{static} information: they all were \emph{state} models, in which the accessibility relations described the agents' uncertainty concerning the current state. The dynamics induced by semi-public actions was simply given by specific \emph{model transformers}. But when dealing with more complex scenarios (involving privacy, secrecy, hacking etc), it is more useful to \emph{represent the actions} themselves in a relational model, with accessibility relations that capture the agents' \emph{uncertainty concerning the current action}. These so-called ``event models'' (or action models) are one of the central features of Dynamic Epistemic Logic \cite{Baltag and Renne:2016,LDII,DitmarschEtAl}, at least in its most popular incarnation (the `BMS approach', due to Baltag, Moss and Solecki \cite{BMS}). Here we adapt this setting to our reading actions.

\smallskip

A \emph{reading event model} is a structure $\bE=(E, \sima, \underline{\bullet})_{a\in A}$, where: $E$ is a finite set of `events'; $\sima\subseteq E\times E$ are equivalence relations, describing each agent's epistemic indistinguishability between events; and $\underline{\bullet}: E \to {\mathcal P}(A)^A$ is a \emph{reading assignment}, associating a reading map $\underline{e}: A \to {\mathcal P}(A)$ to each event $e\in E$.
Intuitively, the events $e\in E$ represent the possible actions that might be taking place at a given moment; $\underline{e}(a)\subseteq A$ is the set of agents whose knowledge bases are accessed (`read') by agent $a$ during action $e$; while the accessibility relations $\sima$ express agent $a$'s
knowledge/beliefs about the current action taking place. As before, the associated reading functions satisfy $a\in \underline{e}(a)$, but in addition they are subject to the constraint
$$e \sima f \mbox{ implies } \underline{e}(a)=\underline{f}(a),$$
saying that \emph{agents know what information bases they read}.
The relations $\sima\subseteq E\times E$ can be extended to groups $B\subseteq A$ and families of groups ${\mathcal B}\subseteq {\mathcal P}(A)$, to define relations $\sim_B$, $\sim^B$ and $\sim^{\mathcal B}$ between events in $E$, in exactly the same way we defined them on states.

\medskip

As usual in Dynamic Epistemic Logic, we describe the \emph{dynamics} induced by a reading event by defining a \emph{product update} operation: a reading action $e$
from a given event model $\bE$ ``acts'' on an input-state $s$ from a given state model $\bS$, producing an output-state $(s,e)$ living in a new
state model $\bS\otimes \bE$ (that represents the possible states and the epistemic uncertainty \emph{after} the event). Once again, we need to adapt this construction to reading actions.

\smallskip
\par\noindent\textbf{Product Update}
Given an epistemic model $\bS= (S, \sima, \underline{\bullet})_{a\in A}$ and a reading event model $\bE=(E, \sima, \underline{\bullet})_{a\in A}$, we can construct their \emph{update product}, which is another epistemic model
$\bS\otimes \bE=(S\times E, \sima, \underline{\bullet})_{a\in A}$, obtained by taking:
\begin{itemize}
\item the set of states is Cartesian product: $S\times E\,\, :=\,\, \{(s,e): s\in S, e\in E\}$.
 \vspace{-0.1cm}
\item the new indistinguishability relations are
$$(s,e)\sima (s',e') \,\, \mbox{ iff } \,\, s\sim_{\underline{e}(a)} s' \mbox{ and } e\sima e'$$
\vspace{-0.1cm}
(which implies that $\underline{e}(a)=\underline{e'}(a)$, and hence that $s\sim_{\underline{e'}(a)} s'$ as well).
\item the truth assignment is as usually inherited from the original state:
$$\underline{(s,e)} \,\, := \,\, \underline{s}.$$ \vspace{-0.3cm}
\end{itemize}

Intuitively, this definition can be justified as follows. The pair $(s,e)$ denotes the output-state produced by performing reading action $e$ on input-state $s$: so our reading events are \emph{deterministic}. The new epistemic relations tell us that: agent $a$'s new knowledge after a reading event $e$ is the result of putting together the knowledge about the original state $s$ gained by reading the information of all agents in $e(a)$ (which incorporates her initial knowledge about $s$, due to the convention $a\in e(a)$) and her knowledge about the event $e$ itself. Finally, the definition of the new truth assignment says that these are \emph{pure reading events}: non-epistemic facts $p$ stay unchanged.

\medskip

\par\noindent\textbf{Drawing conventions} In our graphic representations, we represent the possible events as circles, inside which we write the associated reading map. As before, the epistemic indistinguishability relations between events by links by the respective agent, and as before we skip the loops, directions of arrows, and some arrows obtainable by transitivity.

\begin{exam}\label{public} (\emph{Public Sharing vs. Secret Hacking}) We can represent every semi-public reading/sharing event $!\alpha$, as a \emph{single-event} model $E=\{e\}$, with $\underline{e}=\alpha$. For instance, suppose there are only two agents $A=\{a,b\}$; then the one-event model on the \emph{left} of the diagram below represents the fully public sharing $!a$ by agent $a$ (having only loops for both agents, thus no explicit links in our graph). It is easy to see that taking the product update $S\otimes E$ of any epistemic state model with this event model produces exactly the updated model $\bS^{!a}$.
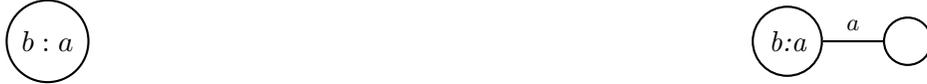
\vspace{-0.5cm}
%
%
%







\begin{center}
      \begin{tikzpicture}[node distance=1.6cm]
        \tikzstyle{zz}=[decorate,decoration={zigzag,post=lineto,post length=5pt}]

        \tikzstyle{w}=[draw=black,thick,circle,
        minimum size=1.6em]

        \tikzstyle{every edge}=[draw,thick,font=\footnotesize]

        \tikzstyle{every label}=[font=\footnotesize]

        \tikzstyle{ev}=[anchor=center,node distance=3.8cm]

        \tikzstyle{wred}=[w,draw=black]



        \node[w,label={above left:},label={above:}] (w1) {$b:a$};

         \node[w,right of =w1, node distance=25em, label={left:},label={above:}] (w2) {b:a};

         \node[w, right of=w2,label={below:}] (w3) {};
------------------------------------

        \path (w2) edge[-] node[above]{$a$} (w3);

      \end{tikzpicture}
\end{center}
\vspace{-0.3cm}
In contrast, the model on the \emph{right} in the above diagram represents the \emph{secret hacking} by $b$ of $a$'s information base. The circle labeled $b:a$ is the actual (hacking) action, while the empty circle is the alternative scenario in which no hacking attack happens (or the attack fails). It is common knowledge that: $a$ has no access to $b$'s data (since she is no hacker); $a$ \emph{doesn't know} that she is being hacked (hence the $a$-link to the empty circle); but she considers this \emph{possible}.
\end{exam}

\begin{exam}\label{mutual-hack} (\emph{Hacking-with-detection vs. Mutual-hacking})
The event model on the \emph{left} in the diagram below represents ``detected hacking'' event (assuming again only two agents $a$ and $b$): everything goes as in the secret-hacking scenario above, except that now $a$ is able to secretly detect the attack (so she knows she is being hacked). The upper $(b:a)$-labeled circle is the actual action (in which the hacking is being detected, so $a$ knows she is being hacked: hence, no $a$-uncertainty links between this circle and any others). Agent $b$ doesn't know that his attack has been detected, but he is of course aware of this possibility (hence the $b$-link between the upper and the lower $(b:a)$-labeled circle, capturing $b$'s uncertainty concerning detection).
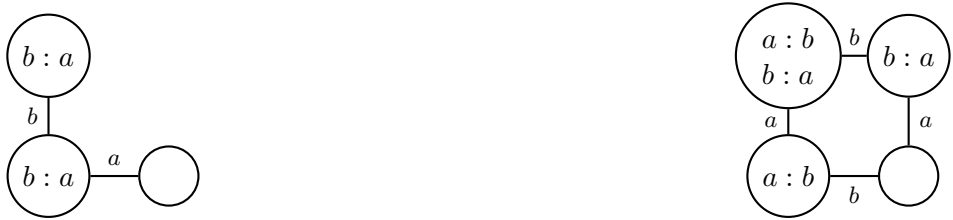
\vspace{-0.5cm}
\begin{center}
      \begin{tikzpicture}[node distance=1.6cm]
        \tikzstyle{zz}=[decorate,decoration={zigzag,post=lineto,post length=5pt}]

          \tikzstyle{w}=[draw=black,thick,circle,
          minimum size=2em]

        \tikzstyle{every edge}=[draw,thick,font=\footnotesize]

        \tikzstyle{every label}=[font=\footnotesize]

        \tikzstyle{ev}=[anchor=center,node distance=3.8cm]

        \tikzstyle{wred}=[w,draw=black]



        \node[w,label={above left:},label={above:},text width=0.7cm] (w5) {$b:a$};

            \node[w, below of=w5,label={below:}] (w6) {$b:a$};

         \node[w, right of=w6,label={above:}] (w7) {};
---------------------------------------------------------------------------

         \node[w,right of =w5, node distance=25em, label={left:},label={above:},text width=0.7cm] (w1) {$a:b$\newline $b:a$};

         \node[w, right of=w1,label={above:}] (w2) {$b:a$};

            \node[w, below of=w1,label={below:}] (w3) {$a:b$};

            \node[w, below of=w2,label={below:}] (w4) {};


            \node[above left of=w1] (ul){ };
            \node[below left of=w3] (dl){ };
            \node[above right of=w2] (ur){ };
            \node[below right of=w4] (dr){ };

\path (w5) edge[-] node[left]{$b$} (w6);

        \path (w6) edge[-] node[above]{$a$} (w7);

        \path (w1) edge[-] node[above]{$b$} (w2);

           \path (w1) edge[-] node[left]{$a$} (w3);

                 \path (w2) edge[-] node[right]{$a$} (w4);

         \path (w3) edge[-] node[below]{$b$} (w4);

      \end{tikzpicture}
      \end{center}
\vspace{-1cm}
The event model on the \emph{right} represents `mutual-secret-hacking': there are only two agents $a$ and $b$, each secretely reading the other's knowledge base. None of them knows that (s)he is being hacked, but (being rational) they consider this possibility. The upper-left circle is the actual event (of double-hacking), while the other circles represent events that are possible according to one agent or another. E.g. the upper-right event represents the case that \emph{only} $b$ is hacking $a$'s database: this is possible according to $b$, hence the horizontal $b$-link between the upper circles.
\end{exam}

\par\noindent\textbf{Adding dynamic modalities for arbitrary reading events}
Given a fixed (locally finite) event model $\bE$,
let $LD\preceq E$ be the logic obtained by adding to $LD\preceq$ dynamic operators $[e]\varphi$ for all events $e\in E$, and let $LDC\preceq E$
be its extension with common knowledge operators. The semantic clause is again given by evaluating $\varphi$ in the updated model:
$$s\models_{\bS}[e]\varphi \,\,\, \mbox{ iff } \,\,\, (s,e)\models_{\bS\otimes \bE}\varphi.$$
The proof of the next result is in Appendix \ref{reduction-completeness}.
\begin{prop}\label{LDE} The dynamic logic $LD\preceq E$ has the same expressivity as its static base $LD\preceq$. A complete axiomatization $\mathbf{LD\preceq E}$ is obtained by adding to the proof system  $\mathbf{LD\preceq}$ the usual axioms and rules of normal modal logic for the dynamic modalities $[e]\varphi$, as well as the following \emph{`Reduction laws' for arbitrary reading events}:
\end{prop}

$$[e ] p \, \, \,  \longleftrightarrow \, \, \,  p
\, \, \, \, \, \, \,  \, \, \, \, \, \, \, \, \, \, \, \, \, \, \,\,\,\, \, \, \,\,\,\, \, \, \, \,\,\,\, \,\,\,\,\, \, \,\, \,
[e ] \neg \varphi \, \, \,  \longleftrightarrow \, \, \, \,\neg [e] \varphi \, \, \, \, \, \, \,  \, \, \, \, \, \, \, \, \, \, \, \, \, \,  \,\,\,\, \, \,\,\,\, \, \, \, \,\,\,\, \,\,\,\,\, \,\,\, \,
[e ] (\varphi \wedge \psi) \, \, \, \longleftrightarrow \, \, \,  [e ] \varphi\wedge [e ] \psi$$
\vspace{-0.5cm}
$$[e](B\preceq C) \,\,\, \longleftrightarrow \,\,\, \underline{e}(B)\preceq \underline{e}(C) \, \, \, \mbox{ for $B\preceq^e C$}  \, \, \, \,  \, \, \, \, \, \, \, \, \, \, \, \, \, \, \, \, \, \, \, \, \, \, \, \, \, \,  \, \, \, \, \, \, \, \, \, \, \, \, [e](B\preceq C) \,\,\, \longleftrightarrow \,\,\, \bot \, \, \, \mbox{ for $B\not\preceq^e C$}$$
$$[e] D_B \varphi \,\,\, \longleftrightarrow \,\,\, \bigwedge \{D_{\underline{e}(B)}[f]\varphi\, : \, f\simB e\}$$
where $B\preceq^e C$ denotes the side condition $\forall f\in E (f\simB e \Rightarrow f\simC e)$
.

\smallskip

\par\noindent\textbf{No reduction laws for $C$ and $Cd$}
Once again, there are no general reduction laws for common knowledge after arbitrary events, nor in fact for common distributed knowledge! To solve this problem, one could again follow the ``second strategy'' (used in Section \ref{Cd}): extend the static base logic, building on ``group epistemic PDL'' \cite{Suzanne}, itself based on \cite{vBvEK}. This would embed common distributed knowledge within a whole range of \emph{distributional levels of knowledge}, similar to \cite{Parikh}, that may be of interest for applications in distributed computing. However, many of the `programs' of epistemic PDL do not seem to have a very transparent and natural epistemic interpretation. Moreover, the resulting reduction laws of (both epistemic PDL in \cite{vBvEK}, and of) group epistemic PDL in \cite{Suzanne} are extremely complex to even state, and too complex to be actually used in any real proofs.

For all these reasons, the ``first strategy'' (used in Section \ref{semi-public}) seems preferable in this case. So in the rest of this section we will follow this strategy, sketching a direct axiomatization of the full dynamic logic of arbitrary events, based on a dynamic analogue of the Induction Rule, that extends the $[!\alpha]C$-Induction Axiom from Section \ref{semi-public} to arbitrary events. Though relatively complex, the resulting rule is still much simpler than reduction laws for epistemic PDL, and can in fact be used in proving various theorems. We leave completeness of this system as a Conjecture, since we did not yet spell out the proof in detail. We plan to do this in a future journal version of this paper.


\medskip

\par\noindent\textbf{Towards an axiomatization of $LDC\preceq E$}
Once again, the idea of the ``first strategy'', when adapted to event models, is to treat the combination $[e]C_B$ \emph{as if} it was a single operator (``common knowledge after event $e$''), like a `dynamic version' of common knowledge. Then one generalizes the Fixed Point and Induction axioms to this dynamic combination, as follows.

The \emph{Dynamic Induction Rule} is a generalization of the $[!\alpha]C$-Induction Axiom to arbitrary events, obtaining by replacing the single premise $\eta$ by a family of premisses $\eta_f$, one for each event $f$ reachable from the given event $e$ by a chain of $B$-links:

\begin{quote}
Given an event $e\in E$, a group $B\subseteq A$, a formula $\varphi$, and a family of formulas $\{\eta_f: f\sim^B e\}$ (one for each event $f\in E$ with $f \sim^B e$), suppose that the formulas
$$\eta_f \,\, \, \to \,\,\, [f]\varphi\wedge D_{\underline{f}(b)}\eta_g$$
are provable, for all $f\sim^B e$, $b\in B$ and $g \sim_b f$. Then we can infer
$$\eta_e \,\, \, \to\,\,\, [e]C_B\varphi.$$
\end{quote}

There is an also a similarly generalized ``Dynamic Fixed Point Axiom'', but that is redundant: it is actually derivable from the usual $C$-Fixed Point Axiom, together with the reduction law for knowledge after $e$.

But the other essential ingredient we need is a Composition Axiom, that allows us to compress strings of dynamic modalities $[e_1]\ldots [e_n]C_B\varphi$ into a single dynamic modality $[e] C_B\varphi$. For this we need to first show that \emph{reading events are closed under sequential composition}.

\smallskip

\par\noindent\textbf{Composition of Event Models} Given full communication event models $\bE=(E, \sima, \underline{\bullet})_{a\in A}$ and $\bE_2=(E_2, \sima, \underline{\bullet})_{a\in A}$,
we can construct their \emph{sequential composition}, which is another event model
$\bE_1; \bE_2=(E_1\times E_2, \sima, \underline{\bullet})_{a\in A}$, obtained by taking:
\begin{itemize}
\item the set of events to be the Cartesian product
$$E_1\times E_2\,\, \,\, :=\,\, \,\, \{e;f: e\in E_1, f\in E_2\},$$
where we used the notation $e;f$ for the ordered pair $(e,f)\in E_1\times E_2$, to stress that it represents the \emph{sequential composition} of the two events.
\item  the epistemic indistinguishability relations to be
$$e; f \, \sima \, e';f' \,\, \,\, \mbox{ iff } \,\, \,\, e\sim_{\underline{f}(a)} e' \mbox{ and } f\sima f'$$
(which implies that $f(a)=f'(a)$, and hence that $e\sim_{\underline{f'}(a)}e'$ as well);
\item the reading assignment function is given by putting
$$\underline{e;f} (a)\,\, \,\, :=\,\,\,\, \underline{e}(\underline{f}(a))=\bigcup_{b\in
\underline{f}(a)} \underline{e}(b).$$
\end{itemize}

\par\noindent\textbf{Closure Under Composition} We can easily see that the function $f: S\times (E_1\times E_2)\to (S\times E_1)\times E_2$ given by
$$f(s,(e; f)):=((s,e), f)$$
is an isomorphism between the models $\bS\otimes (\bE_1; \bE_2)$ and  $(\bS\otimes \bE_1)\otimes \bE_2$.
This establishes the soundness of the following \emph{Event Composition Axiom}
$$[e][f]\varphi \,\, \leftrightarrow [e; f]\varphi$$

We believe that the resulting system is a complete axiomatization of the logic $LDC\preceq E$. Since we did not yet check the proof, we leave this as an open question:

\smallskip

\par\noindent\textbf{Conjecture} A complete axiomatization of $LDC\preceq E$ consists of the following:
\begin{itemize}
\item the axioms and rules of the proof system $\mathbf{LD\preceq E}$ in Proposition \ref{LDE};
\item the axioms and rules for common knowledge;
\item the Necessitation Rule and Distribution Axiom for dynamic modalities $[e]\varphi$;
\item the above Event Composition Axiom;
\item the above Dynamic Induction Rule.
\end{itemize}

\medskip

\par\noindent\textbf{The Idea of the Completeness Proof}
By using the above Reduction Laws (as well as the Necessitation Rule and Distribution Axiom for $[e]$), we can ``push'' dynamic modalities past all the other operators \emph{except for common knowledge} (and eliminate them when they come in front of an atomic proposition $p$ or a comparative statement $B\preceq C$). In this way, we can  reduce any formula in the logic $LCD\preceq E$ to a formula in which \emph{all} dynamic modalities occur \emph{only} in front of common knowledge operators, possibly stacked e.g. in expressions of the form $[e_1][e_2]\ldots [e_n] C_B\varphi$. We can then use the above Composition Law to ``compress'' the stacks into a single dynamic modality
$[e] C_B\varphi$.
Finally, we can deal with the proof theory of expressions of the form $[e] C_B \varphi$ by using the Dynamic Induction Rule (and the Fixed Point Axiom).

\medskip

As mentioned, we are planning to fully settle our Conjecture in a future journal version, by spelling out this proof in detail.

\section{Comparison with other work}

The problem of converting distributed knowledge into common knowledge via sharing was discussed in detail in \cite{Lonely} (where it was shown that this conversion may fail if the agents can share only information expressible by formulas in a given formal language). A more semantic approach was taken in \cite{SB}, based on protocols requiring agents to ``tell everybody all they know'', similarly to our \emph{public sharing actions} $!G$ (but without axiomatizing them).\footnote{In fact, an (unpublished) axiomatization of $!G$-modalities for public sharing events $!G$ (without comparative knowledge, but with a version of common distributed knowledge) was presented by this paper's first author at a workshop affiliated with ESSLLI 2010.}

\smallskip

The more restricted \emph{resolution action} $!(G)$ (by which agents in $G$ share all they know only with each other) was considered in \cite{Agotnes}. The authors gave reduction laws for distributed knowledge after resolution (which can be obtained by applying our general reduction laws to resolution events). They also proposed an axiomatization for the extension with common knowledge, based on a dynamic version of the Induction Rule (as in the second strategy sketched at the end of the
last section). But, as already mentioned, the completeness proof in \cite{Agotnes} contains a gap, and the version of induction rule proposed there seems too weak to be complete. In any case, the strategy pursued in the first part of this paper (adding common distributed knowledge) yields a much simpler complete axiomatization of resolution logic.

\medskip

\emph{Comparative epistemic logic} was introduced in \cite{DitmarschEtAl}, though allowing only \emph{individual} \emph{comparisons} $b\preceq c$ (which the authors write in reverse order, using $c\succeq b$), and combining it only with \emph{individual knowledge} operators $K_a\varphi$. Also, no dynamic extensions were considered. A complete axiomatization of comparative epistemic logic was given in \cite{DitmarschEtAl}, using a non-standard `Gabbay-style' inference rule. Since the rule requires an infinite supply of fresh atomic variables, that completeness proof did not yield decidability. In contrast, our axiomatization immediately gives decidability of this logic (and of its extensions considered in this paper).

\medskip

There is an obvious analogy between some of the axioms in group (III) of Table \ref{tb1} and Armstrong's axioms for functional dependence in Database Theory \cite{Armstrong}, as well as the logical-epistemic properties of variable dependence \cite{BvB2020a,Baltag2016}. This is more than an analogy: one can associate to each agent a corresponding \emph{variable}, taking as ``value'' the agent's information state (her ``local state'', in the sense of \cite{FHMV}). Agent $a$'s associated variable functionally determines agent $b$'s variable iff agent $a$'s information cell at the current state uniquely determines (i.e. it is included in) agent $b$'s information cell, which is the same as epistemic superiority: agent $a$ knows everything known by agent $b$.
Indeed, the fragment $LD\preceq$ of our logic is in a sense just an epistemic reinterpretation of the logic LFD of functional dependence in \cite{BvB2020a} (forthcoming) accompanied by a simplification of the syntax (eliminating the predicates). In this sense, the proof system $\mathbf{LD\preceq}$ for this fragment is not completely new: it is a (simplified) variant of the system $\mathbf{LFD}$ in \cite{BvB2020a}. But our results for all the larger languages are new, as are the setting of semi-public sharing events and the further generalization to arbitrary reading events.

\medskip

We should stress that the completeness and decidability results in this paper are non-trivial: we are not aware of \emph{any} known decidable logic in which our logics can be embedded via some obvious translation. All natural candidates (e.g. the known decidable extensions of mu-calculus or of Propositional Dynamic Logic, the fixed-point extensions of the guarded fragments of First-Order Logic, Monadic Second Order Logic etc.) seem to be able to embed only some proper fragment of our logics. Indeed, the logics presented in this paper are so powerful that they seem to come very close to the borderline where expressivity runs into undecidability.\footnote{Even some very mild extensions (e.g. with dynamic operators $[!\varphi]\psi$) for public announcements in the usual sense) pose problems to our proof methods, and may well turn out to be undecidable.}


\newpage


\appendix

\section{Completeness and decidability of the static logics}\label{CompletenessCd}

In this section, we sketch the proofs of completeness and decidability for the strongest static logic $LCd\preceq$ above, and as an aside indicate how to extract from them similar proofs for its sublogics $LCd$, $LDC\preceq$ and $LD\preceq$.
The proof needs a detour through a more general type of relational models, called \emph{pseudo-models}.

Essentially, pseudo-models treat each group's distributed knowledge relation $\simB$ as a basic, undefined equivalence relation (rather than defining them as intersections of individual knowledge relations); and they also treat comparative knowledge statements $B\preceq C$ as atomic propositions of the usual kind (whose meaning is directly given by truth-assignment functions or valuations, rather than being defined in terms of the relations $\simB$).

In fact, it is convenient to present pseudo-models in the more standard form involving (extended) \emph{valuations}, rather than using truth-assignment functions (although the two presentations are of course equivalent).

\subsection{Soundness and completeness for finite pseudo-models}

\par\noindent\textbf{Pseudo-models} A \emph{pseudo-model} is a structure $\bS=(S, \simB, \|\bullet\|)_{B\subseteq A}$, where: $S$ is a set of states; $\simB\subseteq S\times S$ are binary relations, one for each group $B\subseteq A$; and $\|\bullet\|: Prop\cup \{B\preceq C: B, C\subseteq A\}\to {\mathcal P}(S)$ is an extended valuation function, mapping atomic propositions $p\in Prop$ and formulas $B\preceq C$ into sets of states $\|p\|, \|B\preceq C\|\subseteq S$. These components are required to satisfy the following conditions:
\begin{enumerate}
\item $\simB$ are equivalence relations on $S$;
\item if $s\in \|B\preceq C\|$ and $s\simB t$, then $s\simC t$ and $t\in \|B\preceq C\|$;
\item $\|B\preceq C\|=S$ if $C\subseteq B$;
\item $\|B\preceq C\|\cap \|B\preceq E\|\subseteq \|B\preceq C\cup E\|$;
\item $\|B\preceq C\|\cap \|C\preceq E\|\subseteq \|B\preceq E\|$.
\end{enumerate}

Given a pseudo-model $\bS=(S, \simB, \|\bullet\|)_{B\subseteq A}$, we can define recursively the \emph{satisfaction} relation $s\models\varphi$ between states $s\in S$ and formulas of $LCd\preceq$, by using the valuation on formulas $\theta\in Prop\cup \{B\preceq C: B, C\subseteq A\}$ in the usual way (putting $s\models \theta$ iff $s\in \|\theta\|$), using the standard Tarski clauses for the propositional connectives, and using the standard modal clause for $Cd_{\mathcal B}$ seen as a Kripke modality for the relation
$$\sim^{\mathcal B} \,\, :=\,\, (\bigcup_{B\in {\mathcal B}} \sim_B)^*.$$

\begin{prop}\label{pseudo-soundness}
The axioms and rules of $\mathbf{LCd\preceq}$ are sound with respect to pseudo-models.
\end{prop}

\par\noindent The \emph{proof} is an easy verification: the semantic conditions imposed on pseudo-models are designed to match each of the axioms of $\mathbf{LD\preceq}$, while $Cd$ axioms and rules are always sound for Kripke modalities for relations of the form $(\bigcup_{B\in {\mathcal B}} \sim_B)^*$ based on \emph{any} equivalence relations $\sim_B$.

\medskip

But completeness requires a bit more work.

\medskip

\par\noindent\textbf{Fisher-Ladner Closure} Given any formula $\varphi_0$ in the language of $LCd\preceq$, its Fisher-Ladner closure is the smallest set of formulas $\Sigma=\Sigma(\varphi_0)$ satisfying, for all groups $B, C\subseteq A$, families ${\mathcal B}, {\mathcal C}\subseteq {\mathcal P}(A)$ and formulas $\psi, \theta$:
\begin{enumerate}
\item $\varphi_0\in \Sigma$;
\item $(B\preceq C)\in \Sigma$;
\item if $Cd_{\mathcal B}\psi\in \Sigma$ then $Cd_{\mathcal C}\psi\in \Sigma$;
\item if $Cd_{\mathcal B}\psi \in \Sigma$ and ${\mathcal B}$ is \emph{not} a singleton (consisting of a single set ${\mathcal B}=\{b\})$, then $D_C Cd_{\mathcal B}\psi \in \Sigma$;
\item if $\psi\in \Sigma$ and $\theta$ is a subformula of $\psi$, then $\theta\in \Sigma$;
\item $\Sigma$ is closed under single negations\footnote{The single negation $\sim\varphi$ is defined as: $\sim \varphi:=\theta$ if $\varphi$ is of the form $\neg\theta$; and $\sim\varphi:=\neg\varphi$ if $\varphi$ is \emph{not} of the form $\neg\theta$ (for any $\theta$).} $\sim$: if $\psi\in \Sigma$, then $(\sim\psi)\in \Sigma$.
\end{enumerate}

Note that (given the fact $K_b$, $D_B$ and $C_B$ are in this language just abbreviations) conditions 3 and 4 imply the following closure conditions:
\begin{description}
\item[3'] if $D_B \psi\in \Sigma$ then $D_C \psi\in \Sigma$;
\item[4'] if $C_B\psi\in \Sigma$ then $K_c C_B\psi\in \Sigma$.
\end{description}
For the sublanguages missing the operator $Cd$, conditions 3 and 4 should be skipped, and replaced with condition 3'(only when the operator $D$ belongs to the given sublanguage) and condition 4' (only when $C$ belongs to it). For $LCd$, we have to skip instead condition 2.

One can easily check that \emph{the Fisher-Ladner closure of any formula is finite}.\footnote{Note that, given that $D_B$ is an abbreviation for $Cd_{\{B\}}$, the restriction to non-singleton families in condition 4 is needed to avoid infinite iterations of $D_B$'s.}

\medskip

\par\noindent\textbf{Finite Canonical Pseudo-Model}
For a fixed formula $\varphi_0$, consider the following ``\emph{canonical pseudo-model} for $\varphi_0$'' $\bS^c=(S^c, \simB, \|\bullet\|)$, where: $S^c$ is the set of all maximally consistent theories $T\subseteq \Sigma=\Sigma(\varphi_0)$ (over the finite sublanguage given by the Fisher-Ladner closure of $\varphi_0$);
for $T\in S^c$, $B\subseteq A$, we first put
$$T_B\,\, :=\,\, \{(B\preceq C)\in T: C\subseteq A\} \cap \{(D_C\varphi)\in T: C\subseteq A \mbox{ with } (B\preceq C)\in T\};$$
then the group
group epistemic relations $\simB$ are given by putting, for all $T,W\in S^c$:
$$T\simB W\,\, \mbox{ iff } \,\, T_B=W_B;$$
and the valuation is given by putting, for all $\theta\in Prop\cup \{B\preceq C: B, C\subseteq A\}$:
$$\|\theta\|=\{T\in S^c: \theta\in T\}.$$
It is easy to check that \emph{$S^c$ is a pseudo-model}: $\simB$ are obviously equivalence relations, and the other conditions are ensured by the axioms. It is also clear that \emph{$S^c$ is finite}: since $\Sigma=\Sigma(\varphi_0)$, the number of maximally consistent subtheories is bounded the size of ${\mathcal P}(\Sigma)$, hence $|S^c|\leq 2^{|\Sigma|}$.

For the following result, it is useful to denote by $\pm\varphi$ any of the formulas in the set $\{\varphi,\sim\varphi\}$, and to extend the sets $T_B$ by putting
$$T_B^{\pm} \,\, :=\,\, \{(\pm B\preceq C)\in T:  C\subseteq A\} \cap \{(\pm D_C\varphi)\in T:C\subseteq A \mbox{ with } (B\preceq C)\in T\}.$$
Then we can characterize $\simB$ in terms of one-way inclusion:
$$T\simB W\,\, \mbox{ iff } \,\, T_B^{\pm} \subseteq W.$$

\begin{lem}\label{TruthLemma} (``\emph{Truth Lemma}'')
Given a finite canonical pseudo-model $S^c$ over some Fisher-Ladner closure $\Sigma$, we have for all $\varphi\in \Sigma$:
$$T\models_{\bS^c} \varphi \,\, \mbox{ iff } \,\, \varphi\in T,$$
for every $T\in S^c$.
\end{lem}

\begin{proof} For $T\in S^c$, we will use the notation $\widehat{T}:=\bigwedge T$.
The proof is by induction on the complexity of $\varphi$, in which we treat the inductive case for $D_B\varphi$ (i.e. $Cd_{\{B\}}\varphi$) separately from the one for $Cd_{\mathcal B}\varphi$ with $|{\mathcal B}|>1$.

\medskip

\noindent\textbf{Base cases}: \emph{Atomic propositions} $p\in Prop$ and \emph{comparative assertions} $B\preceq C$ are taken care by our choice of valuation.

\medskip

\noindent\textbf{Inductive cases for Boolean connectives}: these are trivial.

\medskip

\noindent\textbf{Inductive case for $D_B\varphi$}. \emph{Left-to-right}: assume that $T\models D_B\varphi$, and suppose towards a contradiction that $(D_B\varphi)\not\in T$. Take the set
$$W_0=\{\sim\varphi\}\cup T_B^{\pm}$$
(where $T_B^{\pm}$ is the notation introduced earlier).

\smallskip

\par\noindent\emph{Claim}: $W_0$ is consistent.

\smallskip

\emph{Proof of Claim}: Suppose not. Then we have $\vdash \, \widehat{T_B^{\pm}}\to \varphi$. Applying $D$-Necessitation and $D$-Distribution (derivable in our system), we obtain $\vdash \, D_B \widehat{T_B^{\pm}}\to D_B\varphi$. But it is easy to see that we also have $\vdash \, \widehat{T_B^{\pm}}\to D_B \widehat{T_B^{\pm}}$ (which follows from the theorems $(\vdash \, \pm B\preceq C)\to D_B (\pm B\preceq C)$ and
$\vdash \, B\preceq C\to (\pm D_C\varphi \to D_B \pm D_C\varphi)$, derivable in our system from the Interaction Axioms together with the derivable $S5$ laws for $D_B$), and also $\vdash\, \widehat{T}\to  \widehat{T_B^{\pm}}$ (since $T_B^{\pm}\subseteq T$). Putting all these together, we obtain $\vdash\,  \widehat{T}\to D_B\varphi$. Since $(D_B\varphi)\in \Sigma$ and $T$ is maximally consistent subset of $\Sigma$, this gives us $(D_B\varphi)\in T$, which contradicts our assumption that $(D_B\varphi)\not\in T$.

\smallskip

Given the above Claim, we can use the standard Lindenbaum Lemma for our language to construct a maximally consistent subset $W\in S^c$, with $T_B^{\pm}\subseteq W$ and $(\sim\varphi)\in W$.
The first gives us $T\simB W$, and the second gives us $\varphi\not\in W$, and so $W\not\models\varphi$ (by the induction hypothesis), which together contradict the assumption that $T\models D_B\varphi$.

\medskip

\par\noindent\emph{Right-to-left}: Assume that $(D_B\varphi)\in T$. To prove that $T\models D_B\varphi$, let $W\in S^c$ be s.t. $T\simB W$; it is enough to show that $\varphi\in W$.

For this, note that, by the definition of $\simB$ in our canonical pseudo-model,  $(D_B\varphi)\in T$ and $T\models D_B\varphi$ imply $(D_B\varphi)\in W$, which in its turn implies that $\varphi\in W$ (by the provable $S5$ ``axiom's'' for $D$, in particular Truthfulness: $\vdash\, D_b\varphi\to \varphi$).

\medskip

\noindent\textbf{Inductive case for $Cd_{\mathcal B}\varphi$ with $|{\mathcal B}|>1$}. \emph{Left-to-right}:
Assume that $T\models Cd_{\mathcal B}\varphi$. Let
$$S_T^{\mathcal B}\,  := \, \{W\in S^c: \mbox{ there is a chain } T=T^0\sim_{B^1} \ldots \sim_{B^n} T^n=W
\mbox{ with $n\geq 0$ and all $B_i\in {\mathcal B}$}\}$$
We put $\eta\, :=\, \bigvee \{\widehat{W}: W\in S_T^{\mathcal B}\}$.

\smallskip

\par\noindent\emph{Claim 1}: We have $\vdash \, \eta\to \bigwedge_{B\in {\mathcal B}} D_B\eta$.

\smallskip

\emph{Proof of Claim 1}: Suppose not. Then there is some $B\in {\mathcal B}$ s.t. $\eta\wedge \langle D_B \rangle \neg \eta$ is consistent (where $\langle D_B \rangle \theta:= \neg D_B \neg \theta$ is the existential dual of $D_B$). Given the definition of $\eta$, and the easily proven theorem $\vdash \, \bigvee \{\widehat{V}: V\in S^c\}$, this means there exist $W\in S_T^{\mathcal B}$, $V\in S^c - S_T^{\mathcal B}$ such that $\widehat{W}\wedge \langle D_B\rangle \widehat{V}$ is consistent. But this implies that $W\simB V$ (using the definition of $\simB$, $S5$ laws for $D$, and the axioms of Known Superiority and Knowledge Transfer). From this and $W\in S_T^{\mathcal B}$ (together with $B\in {\mathcal B}$ and the definition of $S_T^{\mathcal B}$), we obtain that $V\in S_T^{\mathcal B}$, which contradicts of our above choice of $V\in S^c - S_T^{\mathcal B}$.

\smallskip

\par\noindent\emph{Claim 2}: We have  $\vdash \, \eta\to \varphi$.

\smallskip

\emph{Proof of Claim 2}: From $T\models Cd_{\mathcal B}\varphi$, using the semantics of $Cd$ and the definition of $S_T^{\mathcal B}$, we obtain that $W\models \varphi$ for all $W\in S_T^{\mathcal B}$. By the induction hypothesis, we get that $\varphi\in W$, hence $\vdash \, \widehat{W}\to\varphi$, for all $W\in S_T^{\mathcal B}$. Using the definition of $\eta$, we derive $\vdash\, \eta\to \varphi$, as desired.

\smallskip

Applying now $Cd$-Necessitation to the theorem in Claims 1, we obtain $\vdash \, Cd_{\mathcal B} (\eta\to \bigwedge_{B\in {\mathcal B}} D_B\eta$, which by the $Cd$-Induction Axiom yields $\vdash \, \eta \to Cd_{\mathcal B}\eta$. Combining this with the theorem $\vdash \, Cd_{\mathcal B}\eta\to Cd_{\mathcal B}\varphi$ (obtain from the theorem in Claim 2 by applying $Cd$-Necessitation and $Cd$-Distribution), we obtain
$\vdash \, \eta\to Cd_{\mathcal B}\varphi$. But we also have $\vdash \, \widehat{T}\to \eta$ (since $T\in S_T^{\mathcal B}$). Putting these together, we obtain $\vdash \,  \widehat{T}\to Cd_{\mathcal B}\varphi$, which implies that $(Cd_{\mathcal B} \varphi)\in T$ (since $(Cd_{\mathcal B}\varphi)\in \Sigma$ and $T$ is a maximally consistent subset of $\Sigma$), as desired.

\medskip

\par\noindent\emph{Left-to-right}: Assume that $(Cd_{\mathcal B}\varphi)\in T$. To prove that $T\models Cd_{\mathcal B}\varphi$, let $W\in S^c$ be reachable by some chain $T=T^0\sim_{B^1} T^1 \ldots \sim_{B^n} T^n=W$ for some $n\geq 0$ and $B^1, \ldots, B^n\in {\mathcal B}$; it is enough to show that $W\models \varphi$.

\par\noindent\emph{Claim}: $(Cd_{\mathcal B}\varphi)\in T^k$ for all $1\leq k\leq n$.

\smallskip

\emph{Proof of Claim}: Induction on $k$. For $k=1$, the claim is true by the assumption that $(Cd_{\mathcal B}\varphi)\in T$. For the inductive step: assume that $(Cd_{\mathcal B}\varphi)\in T^k$. From this, using the theorem
$\vdash \, Cd_{\mathcal B}\varphi\to D_{B^k} Cd_{\mathcal B}\varphi$ (which follows from the $Cd$-Fixed Point Axiom), together with $(D_{B^k}Cd_{\mathcal B}\varphi)\in \Sigma$ (by the closure conditions on $\Sigma$), we get that $(D_{B^k}Cd_{\mathcal B}\varphi)\in T^k$ (since $T^k$ is a maximally consistent subset of $\Sigma$).
From this and $T^k \sim_{B^k} T^{k+1}$, we obtain $(D_{B^k}Cd_{\mathcal B}\varphi)\in T^{k+1}$ (by the definition of $\simB$),
hence $(Cd_{\mathcal B}\varphi)\in T^{k+1}$ (using the theorem $\vdash D_{B^k}Cd_{\mathcal B}\varphi \to Cd_{\mathcal B}\varphi$ and the fact that $T^{k+1}$ is maximally consistent); so we proved the claim for $k+1$, as desired.

\smallskip

Applying now the above Claim to $k:=n$, we obtain that $(Cd_{\mathcal B}\varphi)\in T^n=W$, and hence (by the theorem $\vdash \, Cd_{\mathcal B}\varphi \to\varphi$ and the fact that $W$ is maximally consistent) we have $\varphi\in W$, which implies $W\models \varphi$ (by the induction hypothesis), as desired.
\end{proof}

\begin{cor}\label{pseudo-completeness} The axioms and rules of $\mathbf{LCd\preceq}$ are sound and weakly complete with respect to pseudo-models. Moreover, $\mathbf{LCd\preceq}$ has the \emph{finite pseudo-model property}: it is also complete with respect to finite pseudo-models.
\end{cor}

\begin{proof} Soundness was established in Proposition \ref{pseudo-soundness}. Given any consistent formula $\varphi_0$, construct the canonical pseudo-model $\bS^c$ for $\varphi_0$. By Lindenbaum Lemma, there exists some maximally consistent theory $T_0\in \bS^c$ with $\varphi_0\in T_0$. By the Truth Lemma \ref{TruthLemma}, $T_0$ satisfies $\varphi_0$ in $\bS^c$. Since $\bS^c$ is finite, this gives us weak completeness wrt finite pseudo-models (and hence also wrt all pseudo-models).
\end{proof}

\subsection{From pseudo-models to models}

Given a pseudo-model $\bS=(S, \simB, \|\bullet\|)_{B\subseteq A}$, we construct an \emph{associated model}
$\bM= (H, \sima, \underline{\bullet})_{a\in A}$. The construction technique is a variation of modal \emph{unravelling}, making infinitely many copies of each state:

As new set of \emph{states} we take the set $H$ all `histories', i.e. all finite sequences $h=(s_0, B^1, s_1, \ldots, B^n, s_n)$, with $n\geq 0$, $s_o, \ldots, s_n\in S$ and $B^1, \ldots, B^n \subseteq A$ satisfying $s_{k-1} \sim_{B^k} s_k$ for all $k=1,n$. We denote by $last(h):=s_n$ the last state in history $h$, and by $\to_{B}$ the natural \emph{one-step relation} on histories, given by $h\to_B h'$ iff $h'=(h, B, s')$ (with $last(h) \simB s'=last(h')$).

The one-step relations structure $H$ in a \emph{tree-like manner} (or more precisely, a ``rootless tree'', i.e. a \emph{forest}, since there is no unique root): \emph{every two nodes $h, h'$ of this ``rootless tree'' are connected by a unique non-redundant path}.

\smallskip

To make this tree into a model for our language, we define first a \emph{new one-step relation}  $\stackrel{\sim}{\to}_{B}$, incorporating all the one-step relations labeled by groups that are (locally) at least as knowledgeable as $B$:
$$h \stackrel{\sim}{\to}_{B} h' \,\,\, \mbox{ iff } \,\,\, h\to_{B'} h' \mbox{ for some $B'$ with }
last(h)\models B'\preceq B.$$
In particular, for $B=\{b\}$ with $b\in A$, we obtain new one-step relations $\stackrel{\sim}{\to}_b$ for single agents, and then we can go on to define our \emph{indistinguishability relations} $\simb\subseteq H\times H$, by putting
$$\simb \,\, :=\,\, \left( \stackrel{\sim}{\to}_b\cup \stackrel{\sim}{\ot}_b \right)^*,$$
where $\stackrel{\sim}{\ot}_b$ is the converse of $\stackrel{\sim}{\to}_b$, and $R^*$ is the reflexive-transitive closure of $R$. The relation $\simb$ is the smallest equivalence relation that includes $\stackrel{\sim}{\to}_b$.

\smallskip

Finally, we define our \emph{truth-assignment function}, by putting:
$$\underline{h} \,\, :=\, \, \{p\in Prop: last(h)\models_{\bS} p\}=\{ \{p\in Prop: last(h)\in \|p\|\}$$

\medskip

This gives us the \emph{associated model} $\bM=(H, \simb, \underline{\bullet})_{b\in A}$. To compare it with the original pseudo-model, we can \emph{consider this associated model as a pseudo-model} $\bM=(H, \simB, \underline{\bullet})_{b\subseteq A}$, when endowed with the distributed-knowledge relations $\simB$ (defined as usual by taking intersections: $\simB := \bigcap_{b\in B}\simb$) and the additional comparative ``atoms'' $B\preceq C$ (whose valuation is defined to fit the associated model definition: $h\in \|B\preceq C\|$ iff $\forall h'\in H (h\simB h'\Rightarrow h\simC h')$).
It is obvious that the model-based semantics on $\bM$ agrees with this pseudo-model semantics on $\bM$. So we can now directly compare $\bS$ and $\bM$ as pseudo-models.

\medskip

Before doing this, it is useful to give more concrete characterizations of the distributed-knowledge relations $\simB$ in $\bM$.

\begin{lem}\label{b-path}

The following are \emph{equivalent}, for all $b\in A$ and histories $h,h'\in H$:
\begin{enumerate}
\item $h\simb h'$;
\item the non-redundant path from $h$ to $h'$ consists only of steps of the form $h_n {\to}_{B^n} h_{n+1}$, or $h_n {\ot}_{B^n} h_{n+1}$, with $last(h_n)\models B_n\preceq b$.
\end{enumerate}
\end{lem}

\begin{proof}
This should be obvious, given the definition of $\simb$ on histories, and the uniqueness of the non-redundant path from $h$ to $h'$.
\end{proof}

\begin{lem}\label{path}
The following are \emph{equivalent}, for all $B\subseteq A$ and histories $h,h'\in H$:
\begin{enumerate}
\item $h\simB h'$;
\item the non-redundant path from $h$ to $h'$ consists only of steps of the form $h_n {\to}_{B^n} h_{n+1}$, or $h_n {\ot}_{B^n} h_{n+1}$, with $last(h_n)\models B_n\preceq B$.
\end{enumerate}
\end{lem}

\begin{proof}
This follows immediately from the preceding result, using again the uniqueness of the non-redundant path from $h$ to $h'$ (and condition 5 in the definition of pseudo-models).
\end{proof}

Here are some useful properties of the relations $\stackrel{\sim}{\to}_B$ on histories:

\begin{lem}\label{one-step}
If $h\stackrel{\sim}{\to}_B h'$, then we have:
\begin{enumerate}
\item $last(h)\simB last(h')$;
\item $last(h)\in \|B\preceq C\|_{\bS}$ iff $last(h')\in \|B\preceq C\|_{\bS}$;
\item if any of the two equivalent conditions in the previous part hold, then $h \stackrel{\sim}{\to}_C h'$.
\end{enumerate}
\end{lem}

\begin{proof} Assume $h\stackrel{\sim}{\to}_B h'$. By the definition of $\stackrel{\sim}{\to}_B$, this means that $h\to_{B'} h'$ (i.e. $h'=(h, B', last(h')$ with $last(h)\sim_{B'} last (h')$) and $last(h)\models B'\preceq B$ (i.e. $last(h)\in \|B'\preceq B\|_{\bS}$). Putting together $last(h)\sim_{B'} last (h')$, $last(h)\in \|B'\preceq B\|_{\bS}$ and condition 2 in the definition of pseudo-models, we get that $last(h)\simB last(h')$ and $last(h')\in \|B'\preceq B\|_{\bS}$. The first of these immediately yields part 1 of our Lemma.

As for part 2: it follows from the (already proven) part 1 ($last(h)\simB last(h')$) together again with condition 2 in the definition of pseudo-models.

Finally, for part 3: by (the already proven) clause 2, if any of the two conditions in that part holds, then the first one does, i.e. we have $last(h)\in \|B\preceq C\|_{\bS}$. This together with $last(h)\in \|B'\preceq B\|_{\bS}$ gives us that $last(h) \in \|B'\preceq C\|_{\bS}$ (by condition 5 in the definition of pseudo-models). Combining this with the fact that $h\to_{B'} h'$ (and using the definition of $\stackrel{\sim}{\to}_C$), we obtain that $h \stackrel{\sim}{\to}_C h'$, as desired.
\end{proof}

We can now extend these properties to the relation $\simB$ on histories:

\begin{lem}\label{equiv}
If $h\simB h'$, then we have:
\begin{enumerate}
\item $last(h)\simB last(h')$;
\item $last(h)\in \|B\preceq C\|_{\bS}$ iff $last(h')\in \|B\preceq C\|_{\bS}$;
\item if any of the two equivalent conditions in the previous part hold, then $h \simC h'$.
\end{enumerate}
\end{lem}

\begin{proof} We prove the three parts for all pairs of histories $(h,h')$ with $h\simB h'$. The proof is by \emph{induction on the length $N$ of the non-redundant path} from $h$ to $h'$:

\emph{Base case}: $h=h'$. All parts are trivial in this case (given that $\simB$ are equivalence relations).

\emph{Inductive case}: Suppose the non-redundant path from $h$ and $h'$ has length $N+1$, and let us look at the last transition on this path. Given Lemma \ref{path}, this transition can be either of the form
$h_N {\to}_{B^N} h_{N+1}=h'$, or of the form $h_N {\ot}_{B^N} h_{N+1}=h'$, with $last(h_N)\models B_n\preceq B$. Hence, we have either $h_N \stackrel{\sim}{\to}_B h'$ or $h_N \stackrel{\sim}{\ot}_B h'$. Note that the non-redundant path from $h$ to $h_N$ has length $N$. By the induction hypothesis, the pair $(h,h_N)$ satisfies all three parts of our Lemma (with $h'$ replaced by $h_N$). But (using either $h_N \stackrel{\sim}{\to}_B h'$ or $h_N \stackrel{\sim}{\ot}_B h'$, and applying Lemma \ref{one-step}), we can see that the pair $(h_N, h')$ also satisfies all three parts of our Lemma (with $h$ replaced by $h_N$). Putting these two together (and using the transitivity of respectively $\simB$, logical equivalence and $\simC$), we conclude that the pair $(h,h')$ also satisfies all three parts of our Lemma. \end{proof}

Given that pseudo-models are just Kripke models (with relations $\simB$ that happen to be indexed by groups, and having two kind of ``atoms'': $p\in Prop$ and $B\preceq C$ for $B,C\subseteq A$), it is meaningful to ask if the pseudo-models $\bS$ and $\bM$ are \emph{bisimilar}.

\begin{prop}\label{Representation}
Every pseudo-model $bS$ is a p-morphic image of its associated model $\bM$ (seen as a pseudo-model, as explained above). More precisely, the map $last: H\to S$, mapping every history $h\in H$ to its last element $last(h)$, is a surjective p-morphism\footnote{A p-morphism is a functional bisimulation, cf. \cite{BRV}.} from $\bM$ to $\bS$ (seen as Kripke models with basic relations $\simB$ and atoms in $Prop\cup\{B\preceq C: B,C\subseteq A\}$).
\end{prop}

\begin{proof}
It is clear that $last$ is a well-defined function, and that it is surjective: for any $s\in S$, if we just take the history $h_s=(s)$ of length $1$ that has $s$ itself as its root, then we obviously have $last(h_s)=s$.
We check that $last$ satisfies the conditions of a $p$-morphism:

\emph{Atomic preservation for basic atoms} $p\in Prop$ (i.e. $h\in \|p\|_\bM$ iff $last(h)\in \|p\|_\bS$) is immediate (given the way we defined the truth-assignment map in $\bM$).

\emph{Atomic preservation for comparative ``atoms''} $B\preceq C$ (i.e. $h\in \|B\preceq C\|_\bM$ iff $last(h)\in \|B\preceq C\|_\bS$): For the \emph{left-to-right} implication, assume $h\in \|B\preceq C\|_\bM$, i.e. $h\models_{\bM} B\preceq C$. Construct now the history $h':=(h, B, last(h))$, obtained by appending to $h$ a final $B$-transition from $last(h)$ to $last(h)$. We obviously have $h\to_B h'$, thus $h\stackrel{\sim}{\to}_B h'$, hence $h\simB h'$. From this and $h\models_{\bM} B\preceq C$, we obtain that $h\simC h'$ (since $\bM$ is a ``standard'' model, not a pseudo-model). By Lemma \ref{path} and the structure of $h'$, this means we have $h\stackrel{\sim}{\to}_C h'$. Given that $h'=(h, B, last(h)$, this means that $last(h)\models B\preceq C$, i.e. $last(h)\in \|B\preceq C\|_\bS$, as desired.

For the \emph{right-to-left} implication, assume $last(h)\in \|B\preceq C\|_\bS$, i.e. $last(h)\models B\preceq C$. To prove that $h\in \|B\preceq C\|_\bM$, let $h'\in H$ be s.t. $h\simB h'$, and we have to show that $h\simC h'$. But $h\simB h'$ implies $last(h)\simB h'$ (by part 1 in Lemma \ref{equiv}), which together with $last(h)\models B\preceq C$ gives us $h \simC h'$ (by part 3 in in Lemma \ref{equiv}).

\emph{Forth condition}: assume $h\simB h'$, and we need to prove $last(h)\simB last(h')$. This follows by part 1 in Lemma \ref{equiv}.

\emph{Back condition}: assume $last(h)\simB s'$, and we need to show that there exists some $h'\simB h$ with $last(h')=s'$. For this, we can just take $h':=(h,B,s')$.
\end{proof}

\begin{cor}\label{ModalEquivalence}
The same formulas in $LCd\preceq$ are satisfiable in the pseudo-model $\bS$ as in its associated model $\bM$. More precisely, for every history $h\in H$ and every formula $\varphi$ of $LCd\preceq$, we have:
$$h\models_{\bM} \varphi \, \, \mbox{ iff }\,\, last(h)\models_{\bS}\varphi$$
\end{cor}

\begin{proof}
By Proposition \ref{Representation}, the map $last:H\to S$ is a bisimulation between $\bS$ and $\bM$, seen as Kripke models for the language with modalities $D_B$ and additional ``atoms'' $B\preceq C$. Since $LD\preceq$ is just the basic modal language for this vocabulary, formulas in $LD\preceq$ are preserved by $last$ (by the standard results on preservation of modal formulas under bisimulations, cf. \cite{BRV}).
The fact that the addition of $Cd_{\mathcal B}$ to the language maintains this preservation under $last$ follows from the definition of $Cd_{\mathcal B}$ as a modality for the reflexive-transitive closure of the union of all $\simB$'s (which can be seen as an application of the PDL operations of union of relations and reflexive-transitive closure) and the known result that PDL operations are safe for bisimulation \cite{ELD}.
\end{proof}

\medskip

To finish now the proof of Proposition \ref{LCd}, we put together Corollaries \ref{pseudo-completeness} and \ref{ModalEquivalence}, obtaining (weak) \emph{completeness} of $\mathbf{LCd\preceq}$ for our (intended) models. The \emph{decidability} of the logic $LCd\preceq$ follows in the usual way from the fact (cf. Corollary \ref{pseudo-completeness}) that its complete proof system $\mathbf{LCd\preceq}$ is also sound and complete for \emph{finite} pseudo-models, together with the obvious fact that model-checking for $LCd\preceq$ formulas on a finite model is a decidable task.

\medskip

The completeness proofs for the sublogics $LDC\preceq$ and $LD\preceq$ (i.e. Proposition \ref{LCD-completeness}) can be obtained by eliminating from the above proof the steps corresponding to the missing connectives.

\newpage

\section{Completeness and Reduction of Dynamic Logics}\label{reduction-completeness}

We prove this for the logic $LD\preceq E$, and then sketch how the proof can be adapted to $LCD\preceq!$ and $LD\preceq!$.

\begin{lem}\label{soundness-reduction}
The axioms and rules of $\mathbf{LD\preceq E}$ are sound.
\end{lem}

\begin{proof}
This is an easy verification. The reduction laws reflect the construction of the product update. We only give here the proof of soundness for the reduction for $D_B$. We have the following sequence of equivalencies

$s\models [e]D_B \varphi$ iff $(s,e)\models D_B\varphi$ iff $\forall (s',e')\simB (s,e):\, (s',e')\models \varphi$ iff $\forall e'\simB w \, \forall s'\sim_{\underline{e}(B)} s:\, s'\models [e'] \varphi$ iff
$\forall e'\simB w:\, s\models D_{\underline{e}(B)} [e'] \varphi$ iff
$s\models \bigwedge \{ D_{\underline{e}(B)} [e'] \varphi : e'\simB e\}$.
\end{proof}

\begin{lem}\label{one-step reduction}
Let $\theta$ be any ``static'' formula in $LD\preceq$. Then, for every event $e\in E$, there exists some formula $\theta_e$ in the `static' logic $LD\preceq$, s.t.
$$\vdash \, [e]\theta \leftrightarrow \theta_e$$
is provable in $\mathbf{LD\preceq E}$.
\end{lem}

\begin{proof} Induction on the subformula complexity of the static formula $\theta$:

For $\theta:=p$,
the Atomic Reduction Axiom gives us the appropriate formula $\theta_e:=p$.

For $\theta\, :=\, (B\preceq C)$, the corresponding Reduction Axiom gives us $\theta_e\, :=\, \underline(e)(B)\preceq \underline{e}(C)$.

For $\theta\, := \, \neg \psi$, apply the induction hypothesis to $\psi$; so there exists some `static' formula $\psi_e$, such
that $\vdash \, [e]\psi  \leftrightarrow  \psi_e$. Putting this together with the Reduction Axiom for negation, we get $\vdash \, \theta \leftrightarrow \neg\psi_e$ (so we can take $\theta_e:=\neg\psi_e$).

The case $\theta \, :=\, \phi\wedge \psi$ is similar.

For $\theta:= D_B\psi$, we apply the induction hypothesis to $\psi$; hence for every event $f\simB e$, there exists some static formula $\psi_f$ such that $\vdash \, [f]\psi \leftrightarrow  \psi_f$. Putting this together with the Reduction Axiom for $D_B$, we get $\vdash \, \theta \leftrightarrow \bigwedge_{f\simB e} D_{\underline{e}(B)} \psi_f$ (so we can take $\theta_e:= \bigwedge_{f\simB e} D_{\underline{e}(B)} \psi_f$).
\end{proof}

Now we can prove the first part of Proposition \ref{LDE}: the provable co-expressivity of $LD\preceq E$ and $LD\preceq$.

\begin{lem}\label{reduction}
For every formula $\theta$ of the dynamic logic $LD\preceq E$, there exists some formula $\theta'$ of the static language $LD\preceq$, s.t.
$$\vdash \, \theta \leftrightarrow \theta'$$
is provable in $\mathbf{LD\preceq E}$.
\end{lem}

\begin{proof} Induction on the subformula complexity of the dynamic formula $\theta$:

For $\theta:=p$, or $\theta:=(B\preceq C)$, we can take $\theta'=\theta$ (since this is already in $LD\preceq$).

For $\theta\, := \, \neg \psi$, apply the induction hypothesis to $\psi$; so there exists some `static' formula $\psi'$, such
that $\vdash \, \psi  \leftrightarrow  \psi'$. But then we have $\vdash \, \theta \leftrightarrow \neg\psi'$ (so we can take $\theta':=\neg\psi'$).

The cases $\theta \, :=\, \phi\wedge \psi$ is similar.

For $\theta\,:= D_B\psi$, apply the induction hypothesis to $\psi$; so there exists some `static' formula $\psi'$, such that $\vdash \, \psi  \leftrightarrow  \psi'$. By $D_B$-Necessitation and $D_B$-Distribution, we get that
$\vdash \, D_B \psi  \leftrightarrow  D_B\psi'$ (so we can take $\theta':=D_B\psi'$).

For $\theta\,:= [e]\psi$, apply the induction hypothesis to $\psi$; so there exists some `static' formula $\psi'$, such
that $\vdash \, \psi  \leftrightarrow  \psi'$. By $[e]$-Necessitation and $[e]$-Distribution, we get that
$\vdash \, [e] \psi  \leftrightarrow  [e]\psi'$, and by Lemma \ref{one-step reduction} we get another static formula $\psi'_e$, s.t. we have $\vdash \, [e] \psi'  \leftrightarrow \psi'_e$. Putting these together, we get $\vdash \, [e] \psi  \leftrightarrow \psi'_e$ (so we can take $\theta':=\psi'_e$).
\end{proof}

\medskip

Finally, we can now prove Proposition \ref{LDE} (on completeness and co-expressivity of $LD\preceq E$):

\smallskip

\emph{Proof of Proposition \ref{LDE}}: The first part (provable co-expressivity) is already proven (Lemma \ref{reduction}). As for completeness: let $\theta$ be a consistent formula of $LD\preceq E$. By Lemma \ref{reduction}, there exists some $\theta'$ in $LD\preceq$ s.t. $\vdash \, \theta\leftrightarrow \theta'$ is a theorem in $\mathbf{LD\preceq E}$. So $\theta'$ must be consistent (wrt $\mathbf{LD\preceq E}$, hence also) wrt $\mathbf{LD\preceq}$. By the completeness result for $\mathbf{LD\preceq}$ (Proposition \ref{LCD-completeness}), $\theta'$ must be satisfiable at some state $s$ in some epistemic model $\bS$. But then, given the $\mathbf{LD\preceq E}$-theorem  $\vdash \, \theta\leftrightarrow \theta'$ (and the soundness of $\mathbf{LD\preceq E}$), $\theta$ is also satisfiable (at the same state $s$ in the same model).

\bigskip

The completeness and co-expressivity proof for $LCd\preceq!$ (Proposition \ref{LCd!}) is similar: all the above steps are almost identical, except for the reduction laws for $[!\alpha](B\preceq C)$ and $[!\alpha]Cd_{\mathcal B}\varphi$. But these are in fact simpler than the corresponding reduction laws for
$[e](B\preceq C)$ and
$[e] D_B\varphi$, and so both their soundness and the corresponding inductive cases (when proving the analogue of Lemma \ref{one-step reduction}) are easier to check.

\bigskip

Finally, the proof of the analogue results for $LD\preceq!$ (Proposition \ref{LDsharing-completeness}) is similar to the one for $LCd\preceq!$, and in fact even easier: all the steps are identical, except that the reduction law for  $[!\alpha]Cd_{\mathcal B}\varphi$ is replaced by the very similar reduction law for $[!\alpha]D_B\varphi$.

\end{document}